\documentclass[pdflatex,sn-mathphys-num]{sn-jnl}
\usepackage{graphicx}
\usepackage{amsmath,amssymb,amsfonts}
\usepackage{amsthm}
\usepackage{mathrsfs}
\usepackage{mathtools}
\usepackage[title]{appendix}
\usepackage{booktabs}
\usepackage[section]{placeins}
\newcommand{\COP}{\mathrm{COP}}
\newcommand{\SPN}{\mathrm{SPN}}
\newcommand{\Kc}[2]{\mathcal{K}^{(#2)}_{#1}}         
\newcommand{\Cc}[2]{\mathcal{C}^{(#2)}_{#1}}         
\newcommand{\Sym}{\mathcal{S}}
\newcommand{\psd}{\succeq 0}
\newcommand{\R}{\mathbb{R}}
\newcommand{\Q}{\mathbb{Q}}
\newcommand{\N}{\mathbb{N}}
\newcommand{\horn}{H}
\newcommand{\ones}{J}
\newcommand{\eps}{\varepsilon}
\newcommand{\inner}[2]{\langle #1,\, #2\rangle}
\newcommand{\pA}[1]{p_{#1}}
\newcommand{\Mom}{\mathrm{M}}

\theoremstyle{plain}
\newtheorem{theorem}{Theorem}[section]
\newtheorem{lemma}[theorem]{Lemma}
\newtheorem{proposition}[theorem]{Proposition}
\newtheorem{corollary}[theorem]{Corollary}
\theoremstyle{definition}
\newtheorem{definition}[theorem]{Definition}
\newtheorem{example}[theorem]{Example}
\newtheorem{convention}[theorem]{Convention}
\theoremstyle{remark}
\newtheorem{remark}[theorem]{Remark}

\begin{document}
\hypersetup{pdftitle={Explicit Separators for Consecutive Levels of Parrilo's Sum-of-Squares Hierarchy over the Copositive Cone},pdfauthor={Jiachen Shen and Hui Zhong}}

\title[Separators for Parrilo's SOS hierarchy]{Explicit Separators for Consecutive Levels of Parrilo's Sum-of-Squares Hierarchy over the Copositive Cone}

\author[1]{\fnm{Jiachen}\sur{Shen}}\email{jshen28@cougarnet.uh.edu}
\author*[2]{\fnm{Hui}\sur{Zhong}}\email{zhongh7@miamioh.edu}
\affil[1]{\orgname{University of Houston}, \orgaddress{\city{Houston}, \state{Texas}, \country{USA}}}
\affil[2]{\orgname{Miami University}, \orgaddress{\city{Oxford}, \state{Ohio}, \country{USA}}}

\abstract{Parrilo's cones $\Kc{n}{r}$ form a nested sequence of semidefinite-representable inner approximations of the
copositive cone $\COP_n$. For $n=5$ their union is all of $\COP_5$, yet no single level attains it, and
whether consecutive levels actually differ had remained open beyond the classical first step. No explicit
matrix in $\Kc{n}{t}\setminus\Kc{n}{t-1}$ had, to our knowledge, been published for any $t\ge2$ and $n\ge5$. We settle the first three
cases. Explicit rational matrices, obtained from diagonal scalings of the Horn matrix shifted along a
positive interior direction, lie in $\Kc{5}{2}\setminus\Kc{5}{1}$, in $\Kc{5}{3}\setminus\Kc{5}{2}$, and in
$\Kc{5}{4}\setminus\Kc{5}{3}$, giving three consecutive strict inclusions
$\Kc{5}{1}\subsetneq\Kc{5}{2}\subsetneq\Kc{5}{3}\subsetneq\Kc{5}{4}$. Each is certified by an exact rational
Gram matrix and an exact rational dual moment functional, re-verified by a standalone program in integer
arithmetic. The separations are robust. One fixed certificate pair covers an interval of shifts of width
exceeding $3\cdot10^{-3}$, and $\Kc{5}{2}\setminus\Kc{5}{1}$ has nonempty interior. Combining a scaling
theorem of Dickinson, D\"ur, Gijben and Hildebrand with the completeness theorem of Schweighofer and Vargas
shows further that strict adjacent inclusions recur at arbitrarily large levels. All separators were located
by one threshold device: the least shift $\eps_r(M)$ carrying $M$ into $\Kc{5}{r}$ along an interior
direction is nonincreasing in $r$, and each strict drop between levels marks a window of separators.}

\keywords{copositive cone, sum-of-squares hierarchy, Parrilo relaxation, Horn matrix, exact certificates}

\pacs[MSC Classification]{90C22, 90C23, 14P10, 15B48}

\maketitle

\section{Introduction}\label{sec:intro}

A symmetric matrix $A\in\Sym^n$ is copositive if $x^{\top}Ax\ge0$ for every $x$ in the nonnegative orthant.
Optimization over the copositive cone
\[
  \COP_n=\bigl\{A\in\Sym^n : x^{\top}Ax\ge0 \text{ for all } x\ge0\bigr\}
\]
is a linear conic problem that encodes an extraordinary range of hard questions. They include the stability number of a
graph \citep{deKlerkPasechnik2002}, standard quadratic optimization \citep{BomzeDeKlerk2002}, and, through
Burer's reformulation, every mixed-binary quadratic program \citep{Burer2009}. See \citet{Dur2010} for a
survey. All the difficulty sits in the cone, and membership testing is co-NP-complete
\citep{MurtyKabadi1987}. Parrilo \citep{Parrilo2000} introduced the inner approximations that this paper
studies,
\[
  \Kc{n}{r}=\Bigl\{A\in\Sym^n \;:\;
  \bigl(\textstyle\sum_{i=1}^n x_i^2\bigr)^{r}\,\bigl(x^{\circ2}\bigr)^{\top}A\,\bigl(x^{\circ2}\bigr)
  \ \text{is a sum of squares}\Bigr\},
  \qquad r=0,1,2,\dots,
\]
a nested sequence of semidefinite-representable cones with
$\Kc{n}{0}\subseteq\Kc{n}{1}\subseteq\cdots\subseteq\COP_n$. Optimizing over level $r$ in place of $\COP_n$
gives the standard sum-of-squares upper and lower bounds for copositive programs, so the geometry of the
inclusions controls what an extra level of the hierarchy can buy.

The concrete stakes are visible already in the oldest application. Writing the stability number $\alpha(G)$
of a graph, the maximum size of a set of pairwise non-adjacent vertices, as a copositive program \citep{deKlerkPasechnik2002} and relaxing it to level $r$ yields the bounds
$\vartheta^{(r)}(G)$, which decrease to $\alpha(G)$. Whether a higher level improves the bound is, before
anything else, the question of whether the underlying cones actually differ. For the five-cycle the answer at
the first level is the Horn matrix, the matrix around which this entire paper is built
(Section~\ref{sec:stability}). Establishing that the cones separate at each step is the prerequisite for any
statement that the bounds do, and it is exactly this prerequisite that has been missing beyond the first
step.

That geometry has been understood at the ends of the chain for some time, and not at all in the middle. At
the base, $\Kc{n}{0}$ is the cone of sums of a positive semidefinite and a nonnegative matrix, which equals
$\COP_n$ for $n\le4$ \citep{Diananda1962}. The Horn matrix $\horn$ separates the two for $n=5$, and Parrilo
showed $\horn\in\Kc{5}{1}$, so the first inclusion $\Kc{5}{0}\subsetneq\Kc{5}{1}$ is strict with an explicit
witness. At the top, no fixed level equals $\COP_n$ for $n\ge5$ \citep[Corollary~2.3]{DDGH2013}, yet for
$n=5$ the union of all levels is everything: $\COP_5=\bigcup_r\Kc{5}{r}$ \citep{SchweighoferVargas2024}.
Between these facts lies the natural question: is each consecutive inclusion
$\Kc{5}{r}\subseteq\Kc{5}{r+1}$ itself strict? Here essentially nothing was known. Gulati, Nechita and Park
raise the question explicitly \citep{GNP2025}, and Britz and Laurent state that they are not aware of any
explicit matrix in $\Kc{n}{t}\setminus\Kc{n}{t-1}$ for any $t\ge2$ and $n\ge5$ \citep{BritzLaurent2025}. A
hierarchy in active use for two decades thus had exactly one of its steps certified strict.

This paper certifies the next three and describes the structure around them. The witnesses are explicit rational
matrices, and every claim about them reduces to finite rational identities that a short standalone program,
using integer arithmetic only, re-verifies from scratch.

\smallskip
Our contributions are the following.
\begin{enumerate}
\item \textbf{The first open pair is strict} (Theorem~\ref{thm:A}): with
  $D=\mathrm{diag}(\tfrac14,4,1,1,1)$, the matrix $M^{\ast}=D\horn D+\tfrac1{300}\ones$ lies in
  $\Kc{5}{2}\setminus\Kc{5}{1}$. The membership certificate is a rational block Gram matrix. The exclusion
  certificate is a rational moment functional with pairing $-\tfrac{18919}{28800}$.
\item \textbf{So is the second} (Theorem~\ref{thm:Aprime}): with
  $D'=\mathrm{diag}(\tfrac1{16},64,1,1,1)$, the matrix $M'=D'\horn D'+\tfrac1{800}\ones$ lies in
  $\Kc{5}{3}\setminus\Kc{5}{2}$.
\item \textbf{And so is the third} (Theorem~\ref{thm:Adoubleprime}): with
  $D''=\mathrm{diag}(\tfrac1{12},48,1,1,1)$ and $d$ its diagonal, the matrix
  $M''=D''\horn D''+\tfrac1{1024}dd^{\top}$ lies in $\Kc{5}{4}\setminus\Kc{5}{3}$. Together these give
  $\Kc{5}{1}\subsetneq\Kc{5}{2}\subsetneq\Kc{5}{3}\subsetneq\Kc{5}{4}$: three consecutive strict inclusions,
  each with an explicit witness. The third uses a rank-one interior shift $dd^{\top}$ in place of $\ones$,
  the arithmetically better-conditioned member of the family of interior directions in
  Lemma~\ref{lem:interiordir}.
\item \textbf{The first gap is full-dimensional} (Theorem~\ref{thm:B}): a fixed pair of certificates covers the whole segment
  $M_0+\eps\ones$ for $\eps\in[\tfrac1{10000},\tfrac{9981109}{2994257024})$, and
  $\Kc{5}{2}\setminus\Kc{5}{1}$ has nonempty interior in $\Sym^5$.
\item \textbf{Strict steps recur unboundedly} (Theorem~\ref{thm:D}): combining the scaling theorem of
  \citet{DDGH2013} with the completeness theorem of \citet{SchweighoferVargas2024} shows that for every $R$
  some scaling of $\horn$ lies in $\Kc{5}{s}\setminus\Kc{5}{s-1}$ with $s>R$. The chain never stabilizes.
\item \textbf{A search device} (Section~\ref{sec:threshold}): the threshold
  $\eps_r(M)=\min\{\eps\ge0:M+\eps\ones\in\Kc{5}{r}\}$ is computable by one semidefinite program per level,
  is nonincreasing in $r$, and every strict drop between consecutive levels exhibits a window of separators
  (Lemma~\ref{lem:threshold}). All our witnesses were found this way. The same principle applied along the
  rank-one direction $dd^{\top}$ produced the third pair, and the measured threshold landscape
  (Section~\ref{sec:landscape}) locates the window used there.
\end{enumerate}

\smallskip
\emph{Context and related work.} Several strands of prior work meet here.
\emph{Copositive optimization.} Copositive programs capture hard problems exactly. The maximum stable set and
standard quadratic optimization are among them \citep{MurtyKabadi1987,Bomze2000,BomzeDeKlerk2002,Motzkin1965},
and every nonconvex quadratic program with linear and binary constraints has an exact copositive reformulation
\citep{Burer2009}. Tractable inner and outer approximations of $\COP_n$ therefore matter. See the surveys \citep{Dur2010,Bomze2012,Burer2012} and the annotated bibliography
\citep{Bomze2012a}.
\emph{Structure of $\COP_5$ and the Horn orbit.} The completely positive and copositive cones are dual and
classical \citep{Hall1963,Berman2003}. Membership is co-NP-hard \citep{MurtyKabadi1987,Dickinson2014}, with
explicit criteria in low dimensions \citep{Valiaho1986}. Diananda's
equality $\COP_n=\SPN_n$ for $n\le4$, where $\SPN_n$ is the cone of sums of a positive semidefinite and an entrywise
nonnegative matrix \citep{Diananda1962}, and Horn's counterexample fix the picture at $n=5$;
the extreme rays of $\COP_5$ were classified by Baumert and by Hildebrand
\citep{Baumert1966,Hildebrand2012,Hildebrand2014,HaynsworthHoffman1969}, placing the Horn orbit among the
exceptional extremes. Laurent and Vargas characterize exactly which scalings $D\horn D$ lie in $\Kc{5}{1}$
\citep[Theorem~3.8]{LaurentVargas2109}, a criterion we use as a screen and an independent check.
\emph{Sum-of-squares and moment methods.} The hierarchy $\Kc{n}{r}$ is Parrilo's SOS relaxation of copositivity
\citep{Parrilo2000,Parrilo2003}, one instance of the moment/SOS approach to polynomial optimization
\citep{Lasserre2001,Nie2014,Laurent2009,BlekhermanParriloThomas2013,Lasserre2010,deKlerkLaurent2019} resting on
Positivstellens\"atze of Schm\"udgen, Putinar and P\'olya \citep{Schmudgen1991,Putinar1993,PowersReznick2001,
Reznick2000,NieSchweighofer2007}. Linear and semidefinite inner/outer approximations of the copositive cone
were developed by many authors \citep{deKlerkPasechnik2002,Bundfuss2009,Yildirim2012,Gvozdenovic2007,
Dickinson2010}.
\emph{Exactness and convergence.} Dickinson, D\"ur, Gijben and Hildebrand proved the hierarchy is not invariant
under diagonal scaling and that no fixed level captures $\COP_n$ for $n\ge5$ \citep{DDGH2013,Dickinson2014};
Laurent and Vargas reduced $\COP_5=\bigcup_r\Kc{5}{r}$ to the Horn orbit \citep{LaurentVargas2205}, which
Schweighofer and Vargas then settled \citep{SchweighoferVargas2024}. Our Theorem~\ref{thm:D} combines the
second and fourth of these to force unboundedly many strict adjacent steps. Theorems~\ref{thm:A},
\ref{thm:Aprime} and~\ref{thm:Adoubleprime} supply the explicit low-level instances the argument cannot.
\emph{Rational SOS certification.} Our exact-verification methodology descends from Peyrl and Parrilo
\citep{PeyrlParrilo2008}, in the line of exact rational sum-of-squares certificates
\citep{Kaltofen2012,Magron2021}. The boundary obstruction of Section~\ref{sec:method} explains why that
methodology, applied as published, does not certify these raw scalings directly, and the interior shift is our workaround.
\emph{Quantum information.} The cones $\Kc{n}{r}$ and their duals govern a hierarchy of extendibility tests for
bosonic quantum states \citep{GNP2025,BritzLaurent2025,Tura2018}, a copositive counterpart of the
symmetric-extension hierarchies for entanglement detection \citep{Doherty2002,Doherty2004};
Section~\ref{sec:quantum} states what our certificates yield there, and a companion paper \citep{PaperII}
develops it.

\smallskip
\emph{Roadmap.} Section~\ref{sec:prelim} fixes notation, the two cone families, and three lemmas, of which
the exclusion principle (Lemma~\ref{lem:exclusion}) carries all non-membership proofs.
Sections~\ref{sec:thmA}, \ref{sec:thmAprime} and~\ref{sec:thmAdoubleprime} prove the three separation
theorems. Section~\ref{sec:threshold} develops the threshold device, Section~\ref{sec:family} the one-parameter
family, and Section~\ref{sec:thmD} the unbounded-recurrence theorem. Section~\ref{sec:coef} gives the
closed-form separators of the elementary coefficient hierarchy, and Section~\ref{sec:Tpsi} runs the whole
programme on the second exceptional orbit $T(\psi)$. Section~\ref{sec:landscape} reports the measured
threshold landscape, Section~\ref{sec:method} the certification pipeline and the standalone verification
protocol, Section~\ref{sec:stability} the consequences for stability-number bounds, and
Section~\ref{sec:quantum} the quantum reading of the dual certificates. The appendices give the certificates
and document the verifier.

\subsection*{Notation}
$\Sym^n$ is the space of real symmetric $n\times n$ matrices with trace inner product
$\inner{A}{B}=\mathrm{tr}(AB)$; $A\psd$ means positive semidefinite. $\ones$ and $I$ are the all-ones and
identity matrices, $\mathbf{1}$ the all-ones vector. For $x\in\R^n$, $x^{\circ2}=(x_1^2,\dots,x_n^2)^{\top}$,
and for a multi-index $\alpha\in\N^n$, $x^{\alpha}=\prod_ix_i^{\alpha_i}$ with $|\alpha|=\sum_i\alpha_i$.
$z_d(x)$ denotes the column vector of all monomials of degree exactly $d$ in $x_1,\dots,x_5$, ordered
lexicographically; $c_r(A)$ is the coefficient vector of $(\sum_ix_i^2)^r\pA{A}$ in the degree-$(2r+4)$
monomial basis. Matrix indices on $5\times5$ matrices are read modulo $5$. $\Q$-arithmetic statements
(``verified exactly'') mean identities of rational numbers, checked without floating point.

\section{Preliminaries}\label{sec:prelim}

This section fixes the two families of cones the paper studies, the Horn matrix and its diagonal scalings,
and three lemmas that the certificate arguments of Sections~\ref{sec:thmA} to~\ref{sec:family} use repeatedly.
All material here is either classical or elementary. The lemmas are stated in the exact form we need and
proved in full to keep the later arguments self-contained.

\subsection{Copositive matrices and the Parrilo hierarchy}\label{sec:prelim-cones}

Let $\Sym^n$ denote the space of real symmetric $n\times n$ matrices. A matrix $A\in\Sym^n$ is
\emph{copositive} if $x^{\top}Ax\ge 0$ for every $x\in\R^n_{+}$, and the copositive cone is
\begin{equation}\label{eq:cop}
  \COP_n \;=\; \bigl\{A\in\Sym^n : x^{\top}Ax\ge 0 \text{ for all } x\ge 0\bigr\}.
\end{equation}
Verifying membership in $\COP_n$ is co-NP-complete \citep{MurtyKabadi1987}, which motivates tractable inner
approximations. Parrilo's construction \citep{Parrilo2000} replaces the sign constraint by a sum-of-squares
certificate after an even substitution. For $A\in\Sym^n$ write
\begin{equation}\label{eq:pA}
  \pA{A}(x)\;=\;\bigl(x^{\circ 2}\bigr)^{\!\top} A\,\bigl(x^{\circ 2}\bigr)
  \;=\;\sum_{i,j=1}^{n} A_{ij}\,x_i^2x_j^2,
  \qquad x^{\circ 2}:=(x_1^2,\dots,x_n^2)^{\top}.
\end{equation}
Since every nonnegative vector is a coordinatewise square, $A$ is copositive exactly when $\pA{A}$ is
nonnegative on $\R^n$. The Parrilo cone of order $r$ strengthens nonnegativity of $\pA{A}$ to a sum-of-squares
representation after multiplication by a power of $\|x\|^2$:
\begin{equation}\label{eq:parrilo}
  \Kc{n}{r}\;=\;\Bigl\{A\in\Sym^n \;:\;
  \bigl(\textstyle\sum_{i=1}^n x_i^2\bigr)^{r}\,\pA{A}(x)\ \text{is a sum of squares}\Bigr\},
  \qquad r\in\N .
\end{equation}
The hierarchy is nested, $\Kc{n}{r}\subseteq\Kc{n}{r+1}\subseteq\COP_n$: multiplying a sum of squares by
$\sum_i x_i^2$, itself a sum of squares, again yields a sum of squares. At the base level one recovers a
classical object: $\Kc{n}{0}=\Sym^n_{+}+\mathcal{N}^n=:\SPN_n$, the sums of a positive semidefinite and an
entrywise nonnegative matrix \citep{Parrilo2000}. Diananda's theorem gives $\SPN_n=\COP_n$ for $n\le 4$
\citep{Diananda1962}, and Corollary~2.3 of \citet{DDGH2013} shows that no fixed level ever closes the gap for
$n\ge 5$:
\begin{equation}\label{eq:ddgh}
  \COP_n=\Kc{n}{r} \quad\Longleftrightarrow\quad n\le 4 .
\end{equation}
For $n=5$ the hierarchy is nevertheless complete in the pointwise sense: Schweighofer and Vargas
\citep[Theorem~3]{SchweighoferVargas2024} prove
\begin{equation}\label{eq:sv}
  \COP_5=\bigcup_{r\ge 0}\Kc{5}{r},
\end{equation}
so every copositive $5\times 5$ matrix enters the hierarchy at some finite level, while by \eqref{eq:ddgh} no
single level suffices. How the inclusions between \emph{consecutive} cones behave is the subject of this
paper.

\subsection{The Horn matrix and its scalings}\label{sec:prelim-horn}

Our witnesses all live on one classical orbit. Throughout, $C_5$ denotes the $5$-cycle with vertices
$\{1,\dots,5\}$ and edges between cyclically adjacent indices.

\begin{convention}[Horn matrix]\label{conv:horn}
$\horn := 2\bigl(A_{C_5}+I\bigr)-\ones \in\Sym^5$, where $A_{C_5}$ is the adjacency matrix of $C_5$ and
$\ones$ is the all-ones matrix. Explicitly, $\horn_{ij}=1$ if $i=j$ or $i,j$ are adjacent on $C_5$, and
$\horn_{ij}=-1$ otherwise. Indices of $5\times 5$ matrices are read modulo $5$ where convenient.
\end{convention}

Horn's classical example shows $\horn\in\COP_5\setminus\SPN_5$ \citep[p.~25]{Diananda1962}, and Parrilo
verified $\horn\in\Kc{5}{1}$ \citep{Parrilo2000}. The matrix therefore witnesses
$\Kc{5}{0}\subsetneq\Kc{5}{1}$, the only strict consecutive inclusion known before this work.

\begin{convention}[Scalings]\label{conv:scaling}
$\mathcal{D}$ denotes the set of diagonal matrices $D=\mathrm{diag}(d_1,\dots,d_5)$ with all $d_i>0$. For
$D\in\mathcal{D}$ the scaling $D\horn D$ has entries $d_id_j\horn_{ij}$. It is copositive because
$x^{\top}(D\horn D)x=(Dx)^{\top}\horn(Dx)$ and $D$ maps $\R^5_{+}$ onto itself. Copositivity is thus invariant
under $\mathcal{D}$. The cones $\Kc{n}{r}$ are not \citep{DDGH2013}, and this asymmetry is what all our
constructions exploit.
\end{convention}

Membership of a Horn scaling at the first level is understood exactly. We use the following criterion of
Laurent and Vargas. It serves both as a screen for candidate separators and as an independent check on our
dual certificates.

\begin{theorem}[{\citealp[Theorem~3.8]{LaurentVargas2109}}]\label{thm:LV}
Let $D=\mathrm{diag}(d_1,\dots,d_5)\in\mathcal{D}$. Then
\[
  D\horn D\in\Kc{5}{1}
  \quad\Longleftrightarrow\quad
  d_{i-1}d_i+d_id_{i+1}\;\ge\; d_{i-1}d_{i+1}
  \qquad\text{for all } i \pmod 5 .
\]
\end{theorem}

For the two-parameter family $D=\mathrm{diag}(a,b,1,1,1)$ used in Sections~\ref{sec:thmA}
and~\ref{sec:family}, three of the five inequalities hold automatically and the criterion reduces to
\begin{equation}\label{eq:LVab}
  a+ab\ \ge\ b
  \qquad\text{and}\qquad
  b+ab\ \ge\ a .
\end{equation}

\subsection{Three lemmas}\label{sec:prelim-lemmas}

The polynomials in \eqref{eq:parrilo} are homogeneous, and this already pins down which monomials any
sum-of-squares representation can use. The following observation is standard. We record the one-line proof
because the exclusion arguments of Sections~\ref{sec:thmA} to~\ref{sec:thmAdoubleprime} depend on it.

\begin{lemma}[Homogeneous basis]\label{lem:basis}
Let $\sigma\in\R[x_1,\dots,x_n]$ be homogeneous of degree $2d$. If $\sigma=\sum_k g_k^2$, then every $g_k$ is
homogeneous of degree $d$. Consequently, a sum-of-squares representation of
$(\sum_i x_i^2)^{r}\pA{A}$, which is homogeneous of degree $2r+4$, uses only the monomials of degree exactly
$r+2$.
\end{lemma}

\begin{proof}
Write each $g_k$ as a sum of homogeneous components and let $m_k$ (resp.\ $M_k$) be the least (resp.\
greatest) degree occurring in $g_k$. The homogeneous component of degree $2\min_k m_k$ of $\sum_k g_k^2$
equals $\sum_{k:\,m_k=\min} (g_k)_{m_k}^2$, a sum of squares of nonzero polynomials, and is therefore nonzero.
Since $\sigma$ has no component of degree below $2d$, we get $\min_k m_k\ge d$. The symmetric argument at the
top degree gives $\max_k M_k\le d$. Hence every $g_k$ is homogeneous of degree $d$.
\end{proof}

For a matrix argument the polynomial $(\sum_i x_i^2)^{r}\pA{A}$ is even in each variable separately. Even
symmetry lets one search for structured Gram matrices without loss of generality, which is how all
certificates in this paper were found. The proofs of the theorems do not rely on it, but we state the
reduction to make the search reproducible.

\begin{lemma}[Parity averaging]\label{lem:parity}
Let $\sigma$ be a polynomial that is even in each variable and suppose $\sigma=z^{\top}Qz$ with $Q\psd$,
where $z$ is a vector of monomials. For a sign vector $s\in\{\pm 1\}^n$ let $R_s$ be the diagonal matrix
acting on $z$ induced by $x_i\mapsto s_ix_i$. Then
$\bar Q:=2^{-n}\sum_{s}R_s^{\top}QR_s$ is positive semidefinite, satisfies $\sigma=z^{\top}\bar Qz$, and
$\bar Q_{\alpha\beta}=0$ whenever the monomials $z_\alpha,z_\beta$ have different degree parities in some
variable. The Gram matrix may therefore be taken block-diagonal, with one block per parity class of the
monomial basis.
\end{lemma}

\begin{proof}
Each substitution $x_i\mapsto s_ix_i$ fixes $\sigma$ and maps the representation $z^{\top}Qz$ to
$z^{\top}(R_s^{\top}QR_s)z$. Averaging over the $2^n$ sign vectors preserves both positive semidefiniteness
and the represented polynomial. The entry $\bar Q_{\alpha\beta}$ carries the factor
$2^{-n}\sum_s s^{\gamma}$ with $\gamma$ the componentwise parity difference of $z_\alpha$ and $z_\beta$,
and this average vanishes unless $\gamma=0$.
\end{proof}

The last lemma is the engine behind every non-membership claim in the paper. Failing to lie in $\Kc{n}{r}$
means that a specific linear functional separates the target polynomial from the sum-of-squares cone. The
lemma packages this separation in a form that a finite rational object can certify and that a referee can
recheck by hand.

\begin{lemma}[Exclusion principle]\label{lem:exclusion}
Fix $n$, $r$, and let $z$ list all monomials of degree exactly $r+2$ in $x_1,\dots,x_n$. For a symmetric
matrix $A$ let $c_r(A)$ denote the coefficient vector of $(\sum_i x_i^2)^{r}\pA{A}$ in the monomial basis of
degree $2r+4$. Let $y$ be a rational functional on the degree-$(2r+4)$ monomials and define the moment matrix
$\Mom(y)_{\alpha\beta}:=y\bigl(z_\alpha z_\beta\bigr)$. If
\[
  \Mom(y)\psd
  \qquad\text{and}\qquad
  \inner{y}{c_r(A)}<0,
\]
then $A\notin\Kc{n}{r}$.
\end{lemma}

\begin{proof}
Suppose $A\in\Kc{n}{r}$. By Lemma~\ref{lem:basis} there is $Q\psd$, indexed by $z$, with
$(\sum_i x_i^2)^{r}\pA{A}=z^{\top}Qz$. Comparing coefficients monomial by monomial,
\[
  \inner{y}{c_r(A)}
  \;=\;\sum_{\alpha,\beta} y\bigl(z_\alpha z_\beta\bigr)\,Q_{\alpha\beta}
  \;=\;\inner{\Mom(y)}{Q},
\]
the trace inner product of two positive semidefinite matrices. Writing
$\inner{\Mom(y)}{Q}=\bigl\|\,Q^{1/2}\Mom(y)^{1/2}\bigr\|_F^2\ge 0$ contradicts strict negativity.
\end{proof}

\begin{remark}\label{rem:exclusion-scope}
Lemma~\ref{lem:exclusion} quantifies over \emph{all} Gram matrices: no parity or block structure is assumed
of $Q$. In our certificates the functional $y$ is supported on the monomials that are even in each variable
and set to zero elsewhere, which makes $\Mom(y)$ block-diagonal with respect to the parity classes of
Lemma~\ref{lem:parity}. Positive semidefiniteness of $\Mom(y)$ is then verified exactly, as a single matrix,
by the standalone checker described in Section~\ref{sec:method}.
\end{remark}

\subsection{A worked instance}\label{sec:prelim-worked}

Before turning to the new separators it helps to see the machinery on the one case that is classical, the
membership $\horn\in\Kc{5}{1}$. Here $\pA{\horn}=\sum_{i,j}\horn_{ij}x_i^2x_j^2$, and the level-one target is
$(\sum_ix_i^2)\pA{\horn}$, a degree-six form that Parrilo \citep{Parrilo2000} showed to be a sum of squares.
Our pipeline reproduces this with an exact rational certificate: a block Gram matrix $Q_\horn$ over the $35$
monomials of degree three, splitting by Lemma~\ref{lem:parity} into five
$5\times5$ blocks and ten scalars, with $(\sum_ix_i^2)\pA{\horn}=z_3^{\top}Q_\horn z_3$ verified coefficient by
coefficient over $\Q$. On the exclusion side, the dual search of Section~\ref{sec:method} run on $\horn$
returns no witness, in agreement with membership: the semidefinite program certifying $\horn\notin\Kc{5}{1}$
is infeasible, and its infeasibility is itself certified by the primal certificate $Q_\horn$. This example
fixes the two halves of every later argument, a primal Gram matrix for membership and a dual moment functional
for exclusion, on a case whose answer is already known, and it is the first of the three consistency controls
listed in Section~\ref{sec:method-verify}.

The zeros of $\pA{\horn}$ that Proposition~\ref{prop:kernel} will use are already visible here. The Horn form
$u^{\top}\horn u$ vanishes on the nonnegative orthant at the five minimal zeros $e_i+e_j$ indexed by the
non-adjacent pairs $\{1,3\},\{1,4\},\{2,4\},\{2,5\},\{3,5\}$ of $C_5$, where $\horn$ restricts to
$\left(\begin{smallmatrix}1&-1\\-1&1\end{smallmatrix}\right)$. We use these five minimal rays; the
nonnegative zero set is larger, as Remark~\ref{rem:minimalzeros} records. Correspondingly $Q_\horn$ is singular,
with the five evaluation vectors $z_3(e_i+e_j)$ in its kernel. The scalings studied below move these zeros around by
$D^{-1}$ but never remove them, which is the source of the boundary phenomenon that makes the new certificates
delicate.

\section{An explicit separator for the first open pair}\label{sec:thmA}

This section proves $\Kc{5}{1}\subsetneq\Kc{5}{2}$ by exhibiting one rational matrix together with two finite
rational certificates: a Gram matrix that places it inside $\Kc{5}{2}$, and a moment functional that, through
Lemma~\ref{lem:exclusion}, keeps it outside $\Kc{5}{1}$. The candidate is a Horn scaling pushed slightly
along the all-ones matrix. The scaling is chosen to violate the criterion of Theorem~\ref{thm:LV}. The shift
along $\ones$ moves the matrix off the boundary of $\Kc{5}{2}$, where rational Gram matrices exist and can be
written down (Section~\ref{sec:method} explains why the unshifted scaling resists exact certification). How
the pair $(a,b)$ and the shift $1/300$ were located is the subject of Section~\ref{sec:threshold}. Nothing in
the present proof depends on that search.

\begin{theorem}\label{thm:A}
Let $D=\mathrm{diag}\bigl(\tfrac14,\,4,\,1,\,1,\,1\bigr)$, $M_0=D\horn D$ and
\[
  M^{\ast}\;=\;M_0+\tfrac{1}{300}\,\ones .
\]
Then $M^{\ast}\in\Kc{5}{2}\setminus\Kc{5}{1}$. In particular $\Kc{5}{1}\subsetneq\Kc{5}{2}$.
\end{theorem}

\begin{proof}
\emph{Membership in $\Kc{5}{2}$.} Appendix~\ref{app:certA} specifies a rational symmetric matrix $Q$, given in
full in the accompanying archive, indexed by
the $70$ monomials of degree $4$ in five variables and block-diagonal with respect to the $16$ parity classes
of Lemma~\ref{lem:parity} (one $15\times 15$ block, ten $5\times 5$ blocks, five $1\times 1$ blocks),
satisfying the polynomial identity
\begin{equation}\label{eq:gramA}
  z_4(x)^{\top}\,Q\,z_4(x)\;=\;\Bigl(\sum_{i=1}^5 x_i^2\Bigr)^{2}\,\pA{M^{\ast}}(x)
\end{equation}
identically over $\Q$, with every block of $Q$ positive semidefinite. Both facts are checked in exact
rational arithmetic. Identity~\eqref{eq:gramA} is confirmed coefficient by coefficient, and positive
semidefiniteness of each block by an exact $LDL^{\top}$ decomposition (Section~\ref{sec:method}). The
right-hand side of \eqref{eq:gramA} is therefore a sum of squares, so $M^{\ast}\in\Kc{5}{2}$. A single
coefficient shows the shape of the check: the monomial $x_1^8$ arises in $z_4^{\top}Qz_4$ only from the
diagonal entry of $Q$ at the basis monomial $x_1^4$, and in the right-hand side only from
$M^{\ast}_{11}=\tfrac{1}{16}+\tfrac{1}{300}=\tfrac{79}{1200}$. The archived $Q$ has exactly
$Q_{x_1^4,x_1^4}=\tfrac{79}{1200}$, and the remaining $69$ even-monomial coefficients match in the same way.
To make the certificate concrete, one of the ten $5\times5$ blocks, in the basis
$(x_4x_5^3,\ x_4^3x_5,\ x_3^2x_4x_5,\ x_2^2x_4x_5,\ x_1^2x_4x_5)$, is
\begin{equation}\label{eq:blockA}
  \setlength\arraycolsep{3pt}
  \begin{pmatrix}
  \tfrac{337}{750} & \tfrac{1127}{2700} & \tfrac{99}{850} & -\tfrac{17417}{10200} & \tfrac{299}{850}\\[2pt]
  \tfrac{1127}{2700} & \tfrac{649}{1500} & \tfrac{599}{5100} & -\tfrac{5749}{3400} & \tfrac{3313}{10200}\\[2pt]
  \tfrac{99}{850} & \tfrac{599}{5100} & \tfrac{14261}{2550} & -\tfrac{5357}{11100} & \tfrac{5549}{5550}\\[2pt]
  -\tfrac{17417}{10200} & -\tfrac{5749}{3400} & -\tfrac{5357}{11100} & \tfrac{2911}{425} & -\tfrac{31619}{22200}\\[2pt]
  \tfrac{299}{850} & \tfrac{3313}{10200} & \tfrac{5549}{5550} & -\tfrac{31619}{22200} & \tfrac{30929}{10200}
  \end{pmatrix}\succeq 0,
\end{equation}
whose positive semidefiniteness is decided by an exact $LDL^{\top}$ over $\Q$. The full set of blocks is in
the archive described in Appendix~\ref{app:certA}.

\emph{Exclusion from $\Kc{5}{1}$.} Appendix~\ref{app:certA} also lists a rational functional $y$ on the
monomials of degree $6$, supported on the monomials that are even in every variable and equal to zero on all
others. Its moment matrix $\Mom(y)$, indexed by the $35$ monomials of degree $3$, is verified to be positive
semidefinite as a full $35\times 35$ matrix, again by exact $LDL^{\top}$. The pairing with the target
evaluates exactly to
\begin{equation}\label{eq:pairingA}
  \inner{y}{c_1(M^{\ast})}\;=\;-\,\frac{18919}{28800}\;<\;0 .
\end{equation}
Lemma~\ref{lem:exclusion} with $r=1$ yields $M^{\ast}\notin\Kc{5}{1}$.
\end{proof}

The two halves of the proof are independent finite computations, and each is re-checked by a standalone
program that reads only the archived certificate, re-derives $\horn$ and $M_0$ from
Conventions~\ref{conv:horn} and~\ref{conv:scaling}, and uses integer and rational arithmetic exclusively
(Section~\ref{sec:method}). The certificate can be re-checked independently by rerunning that program, or by
reverifying \eqref{eq:gramA} to \eqref{eq:pairingA} in any computer-algebra system.

\begin{remark}[Consistency with the level-one criterion]\label{rem:LVcheck}
The unshifted scaling already fails the first level. With $(a,b)=(\tfrac14,4)$, the first inequality of
\eqref{eq:LVab} reads $\tfrac14+1\ge 4$ and is false, so $M_0\notin\Kc{5}{1}$ by Theorem~\ref{thm:LV}. The
theorem does not apply to $M^{\ast}$, which is not a scaling of $\horn$. This is precisely why the dual
certificate \eqref{eq:pairingA} is needed. The same functional $y$ also pairs negatively with $M_0$,
$\inner{y}{c_1(M_0)}=-\tfrac{9981109}{384}$, giving a certificate for $M_0\notin\Kc{5}{1}$ that is
independent of Theorem~\ref{thm:LV}. The agreement between the two routes is a useful check on the entire
pipeline.
\end{remark}

\begin{remark}[Genuine sums of squares are needed]\label{rem:polya}
A simpler certificate class fails on $M^{\ast}$. Call $A$ \emph{coefficientwise certified} at order $r$ if
$(\sum_i x_i^2)^{r}\pA{A}$ has only nonnegative coefficients, a P\'olya-type condition that implies
membership in $\Kc{n}{r}$ through the diagonal Gram matrix. The coefficient of $x_1^2x_2^2x_4^2x_5^2$ in
$(\sum_i x_i^2)^{2}\pA{M^{\ast}}$ equals $-\tfrac{598}{25}$, so no order-two coefficientwise certificate
exists for $M^{\ast}$: the off-diagonal blocks of $Q$ in \eqref{eq:gramA} carry irreducible cancellation.
\end{remark}

\begin{remark}[On the size of the shift]\label{rem:shift-size}
The shift $\tfrac{1}{300}$ is not tuned finely. Section~\ref{sec:family} shows that every
$\eps\in[\tfrac{1}{10000},\,\tfrac{9981109}{2994257024})$ works, an interval of width exceeding
$3\cdot 10^{-3}$, and that the certificates deform explicitly with $\eps$.
\end{remark}

\section{The second pair}\label{sec:thmAprime}

The construction of Section~\ref{sec:thmA} is not specific to the first level. Moving to a more extreme
scaling produces a matrix outside $\Kc{5}{2}$, and the same shift-and-certify pattern then separates the
next pair of cones. Beyond settling a second case, this repetition is evidence that the method, and not an
accident of level one, is doing the work; Section~\ref{sec:landscape} quantifies how the relevant thresholds
vary across scalings and levels.

\begin{theorem}\label{thm:Aprime}
Let $D'=\mathrm{diag}\bigl(\tfrac{1}{16},\,64,\,1,\,1,\,1\bigr)$, $M_0'=D'\horn D'$ and
\[
  M'\;=\;M_0'+\tfrac{1}{800}\,\ones .
\]
Then $M'\in\Kc{5}{3}\setminus\Kc{5}{2}$. In particular $\Kc{5}{2}\subsetneq\Kc{5}{3}$, and combined with
Theorem~\ref{thm:A},
\[
  \Kc{5}{1}\;\subsetneq\;\Kc{5}{2}\;\subsetneq\;\Kc{5}{3}.
\]
\end{theorem}

\begin{proof}
The structure is that of Theorem~\ref{thm:A}, one level up. We record what changes. For membership,
Appendix~\ref{app:certAprime} specifies a rational block Gram matrix, given in full in the accompanying archive, $Q'$ indexed by the $126$ monomials of
degree $5$, block-diagonal with respect to their $16$ parity classes, satisfying
\begin{equation}\label{eq:gramAprime}
  z_5(x)^{\top}\,Q'\,z_5(x)\;=\;\Bigl(\sum_{i=1}^5 x_i^2\Bigr)^{3}\,\pA{M'}(x)
\end{equation}
identically over $\Q$, with every block positive semidefinite. Hence $M'\in\Kc{5}{3}$. For exclusion,
Appendix~\ref{app:certAprime} lists a rational functional $y'$ on the monomials of degree $8$, supported on
the monomials even in every variable and zero elsewhere, whose full $70\times 70$ moment matrix over the
degree-$4$ basis is positive semidefinite, with
\begin{equation}\label{eq:pairingAprime}
  \inner{y'}{c_2(M')}\;=\;-\,\frac{10139181}{25600}\;<\;0 .
\end{equation}
Lemma~\ref{lem:exclusion} with $r=2$ gives $M'\notin\Kc{5}{2}$. All identities and semidefiniteness claims
are verified exactly, by the same standalone checker as in Section~\ref{sec:thmA}.
\end{proof}

\begin{remark}[Why a more extreme scaling is necessary]\label{rem:whyextreme}
The scaling of Theorem~\ref{thm:A} cannot separate the second pair. Its second-level threshold vanishes
numerically, $\eps_2(M_0)=0$ for $M_0=D\horn D$ with $(a,b)=(\tfrac14,4)$ (the profile~\eqref{eq:measured}),
so along $\ones$ there is no window $(\eps_3(M_0),\eps_2(M_0))$ from which to separate $\Kc{5}{2}$ and
$\Kc{5}{3}$ near it. The more extreme scaling $D'$ is chosen, again from its measured threshold profile,
outside $\Kc{5}{2}$, and the exclusion certificate \eqref{eq:pairingAprime} confirms the choice exactly. A
byproduct is worth recording: $M'$ is copositive, so \eqref{eq:pairingAprime} also gives an
explicit exact witness for $\Kc{5}{2}\ne\COP_5$, a fact otherwise available only through the scaling
theorem of \citet{DDGH2013}.
\end{remark}

\section{The third pair}\label{sec:thmAdoubleprime}

The construction is not tied to the all-ones direction. Any matrix $B=DCD$ with $D$ a positive diagonal and
$C$ symmetric with strictly positive entries is, like $\ones$, an interior direction of every $\Kc{5}{r}$
(Lemma~\ref{lem:interiordir}). Shifting a boundary scaling along such a $B$ and certifying the two sides is
the same argument. At the third pair the all-ones window is too narrow for our recovery to resolve at manageable denominators, and a
direction adapted to the zeros of the underlying Horn form does the work.

\begin{theorem}\label{thm:Adoubleprime}
Let $D''=\mathrm{diag}\bigl(\tfrac{1}{12},\,48,\,1,\,1,\,1\bigr)$ with diagonal
$d=(\tfrac1{12},48,1,1,1)$, let $M_0''=D''\horn D''$, and set $B=dd^{\top}=D''\ones D''$. Then
\[
  M''\;=\;M_0''+\tfrac{1}{1024}\,B \;\in\;\Kc{5}{4}\setminus\Kc{5}{3}.
\]
In particular $\Kc{5}{3}\subsetneq\Kc{5}{4}$, and combined with Theorems~\ref{thm:A} and~\ref{thm:Aprime},
\[
  \Kc{5}{1}\;\subsetneq\;\Kc{5}{2}\;\subsetneq\;\Kc{5}{3}\;\subsetneq\;\Kc{5}{4}.
\]
\end{theorem}

\begin{proof}
The structure is that of the previous two theorems, one level higher. For membership,
Appendix~\ref{app:certAdoubleprime} specifies a rational block Gram matrix, given in full in the accompanying archive, $Q''$ indexed by the $210$ monomials
of degree $6$, block-diagonal with respect to their $16$ even parity classes (one of size $35$, ten of size
$15$, five of size $5$), satisfying
\begin{equation}\label{eq:gramAdoubleprime}
  z_6(x)^{\top}\,Q''\,z_6(x)\;=\;\Bigl(\sum_{i=1}^5 x_i^2\Bigr)^{4}\,\pA{M''}(x)
\end{equation}
identically over $\Q$, with every block positive semidefinite. Hence $M''\in\Kc{5}{4}$. The Gram matrix is
obtained by reserving part of the shift: writing $F_{\eps}=(\sum_i x_i^2)^{4}\pA{M_0''+\eps B}$, one
constructs an exact Gram for $F_{1/2048}$ and adds $\tfrac1{2048}G_B$, where
$G_B=\mathrm{diag}(b_\alpha)$ is the strictly positive coefficient-diagonal Gram of
$(\sum_i x_i^2)^{4}\pA{B}$ from Lemma~\ref{lem:interiordir}. Because $G_B$ represents
$(\sum_i x_i^2)^{4}\pA{B}$ exactly, the sum is an exact Gram for $F_{1/1024}=F_{1/2048}+\tfrac1{2048}\,(\sum_i x_i^2)^4\pA{B}$.
For exclusion, Appendix~\ref{app:certAdoubleprime} provides a rational functional $y''$ on the monomials of
degree $10$, supported on those even in every variable and zero elsewhere, whose full $126\times126$ moment
matrix over the degree-$5$ basis is positive semidefinite, with
\begin{equation}\label{eq:pairingAdoubleprime}
  \inner{y''}{c_3(M'')}\;<\;0 .
\end{equation}
Lemma~\ref{lem:exclusion} with $r=3$ gives $M''\notin\Kc{5}{3}$. All identities and semidefiniteness claims
are verified exactly by the same standalone checker as in Section~\ref{sec:thmA}. The exact rational value of
the pairing \eqref{eq:pairingAdoubleprime} is displayed in Appendix~\ref{app:certAdoubleprime}.
\end{proof}

\begin{remark}[Why the all-ones direction is replaced at this level]\label{rem:whyDJD}
Along $\ones$ the third-pair window $\bigl(\eps_4(M_0''),\eps_3(M_0'')\bigr)$ is too thin to contain a
rational shift with manageable denominators. The direction $B=dd^{\top}$ equalizes the values of $\pA{B}$ at
the five minimal zeros of the Horn form $\pA{M_0''}$, which conditions both certificates: the membership Gram
rationalizes on a rounding grid of denominator $27720$ (its reduced entries carrying larger denominators, as
Remark~\ref{rem:lattice-trap} warns), and the exclusion functional on a larger but still explicit grid.
This is the sense in which the separators come from the method, and the choice of interior direction only
controls the arithmetic. The coefficient-diagonal Gram $G_B$ plays here the role that the diagonal Gram of
$\ones$ plays in Lemma~\ref{lem:Jinterior}.
\end{remark}

\section{The $\ones$-threshold device}\label{sec:threshold}

The three separation theorems so far answer a question of the form ``does the gap between two consecutive
cones contain an explicit point?'' This section isolates the search principle that produced those points. The idea is to
attach to every matrix a one-dimensional profile, its distance-to-membership along the all-ones direction,
and to read off separators from strict drops of that profile across levels. The device reduces a search in
the $15$-dimensional space $\Sym^5$ to the comparison of finitely many scalars, each computable by one
semidefinite program.

The all-ones matrix is the right direction because its associated polynomial is as explicit as a sum of
squares can be.

\begin{lemma}[$\ones$ is an interior direction]\label{lem:Jinterior}
For every $n$ and $r$,
\begin{equation}\label{eq:Jidentity}
  \Bigl(\sum_{i=1}^n x_i^2\Bigr)^{r}\pA{\ones}(x)
  \;=\;\Bigl(\sum_{i=1}^n x_i^2\Bigr)^{r+2}
  \;=\;\sum_{|\alpha|=r+2}\binom{r+2}{\alpha}\,\bigl(x^{\alpha}\bigr)^2 ,
\end{equation}
where $\alpha$ ranges over the multi-indices of total degree $r+2$ and $\binom{r+2}{\alpha}$ is the
multinomial coefficient. The Gram matrix of the right-hand side in the monomial basis of degree $r+2$ is
diagonal with strictly positive entries. Consequently $\ones$ lies in the interior of $\Kc{n}{r}$ for every
$r$.
\end{lemma}

\begin{proof}
$\pA{\ones}(x)=\sum_{i,j}x_i^2x_j^2=(\sum_i x_i^2)^2$, which gives the first equality. The second is the
multinomial expansion of $(\sum_i x_i^2)^{r+2}$ after grouping $x^{2\alpha}=(x^{\alpha})^2$. A diagonal Gram
matrix with positive entries is positive definite on the full monomial basis, and a polynomial admitting a
positive definite Gram matrix retains a positive semidefinite one under every sufficiently small
perturbation of its coefficients. Since $A\mapsto$ coefficients of $(\sum_i x_i^2)^r\pA{A}$ is affine, a
neighbourhood of $\ones$ in $\Sym^n$ maps into such perturbations, so $\ones\in\mathrm{int}\,\Kc{n}{r}$.
\end{proof}

The same computation identifies a whole family of interior directions, used in Section~\ref{sec:thmAdoubleprime}.

\begin{lemma}[Positive shifts are interior directions]\label{lem:interiordir}
Let $D=\mathrm{diag}(d_1,\dots,d_n)$ with $d_i>0$ and let $C\in\Sym^n$ have strictly positive entries. Then
$B=DCD$ lies in the interior of $\Kc{n}{r}$ for every $r$, and the Gram matrix of
$(\sum_i x_i^2)^{r}\pA{B}$ in the degree-$(r+2)$ monomial basis may be taken diagonal with strictly positive
entries
\[
  b_\alpha\;=\;\sum_{i,j}d_id_jC_{ij}\binom{r}{\alpha-e_i-e_j},\qquad |\alpha|=r+2,
\]
where $\binom{r}{\beta}$ denotes the multinomial coefficient, zero unless $\beta\in\N^n$. The all-ones case
$\ones=I\ones I$ of Lemma~\ref{lem:Jinterior} and the rank-one case $dd^{\top}=D\ones D$ are both of this form.
\end{lemma}

\begin{proof}
$\pA{B}(x)=\sum_{i,j}d_id_jC_{ij}\,x_i^2x_j^2$ is a strictly positive combination of the monomials
$x_i^2x_j^2$, so multiplying by $(\sum_i x_i^2)^{r}$ and grouping $x^{2\alpha}=(x^{\alpha})^2$ gives
$\sum_{|\alpha|=r+2}b_\alpha (x^{\alpha})^2$ with the stated $b_\alpha$, each a sum of nonnegative terms and
positive because some $i,j$ with $\alpha\ge e_i+e_j$ contributes. A diagonal Gram matrix with positive
entries is positive definite, and the perturbation argument of Lemma~\ref{lem:Jinterior} places $B$ in
$\mathrm{int}\,\Kc{n}{r}$.
\end{proof}

\begin{definition}[Threshold along an interior direction]\label{def:threshold}
Fix a direction $B$ lying in the interior of every level, $B\in\bigcap_{s\ge0}\mathrm{int}\,\Kc{n}{s}$ (by
Lemma~\ref{lem:interiordir} every $B=DCD$ with $D$ positive diagonal and $C$ entrywise positive qualifies, in
particular $\ones$ and $dd^{\top}$). For $M\in\Sym^n$ and $r\in\N$,
\[
  \eps_{r,B}(M)\;:=\;\min\bigl\{\eps\ge 0 \,:\, M+\eps B\in\Kc{n}{r}\bigr\},
  \qquad
  \eps_r(M):=\eps_{r,\ones}(M).
\]
The default direction is $B=\ones$, written $\eps_r$. Section~\ref{sec:thmAdoubleprime} and the fourth-level
profiles of Section~\ref{sec:landscape} use $B=dd^{\top}$, written $\eps_{r,B}$.
\end{definition}

The next lemma checks that the definition makes sense and behaves monotonically. Its third part is the
entire search principle. It is stated for $B=\ones$. The same statements hold verbatim for any interior
direction $B$, since only $B\in\mathrm{int}\,\Kc{n}{r}$ and the convexity and closedness of $\Kc{n}{r}$ are
used.

\begin{lemma}[Threshold calculus]\label{lem:threshold}
Let $M\in\Sym^n$.
\begin{enumerate}
\item The set $\{\eps\ge0 : M+\eps\ones\in\Kc{n}{r}\}$ is a nonempty closed interval $[\eps_r(M),\infty)$ for
  every $M\in\Sym^n$. In particular the minimum in Definition~\ref{def:threshold} is attained and finite.
\item $\eps_{r+1}(M)\le\eps_r(M)$ for every $r$.
\item If $\eps_{r+1}(M)<\eps_r(M)$, then for every
  $\eps\in\bigl(\eps_{r+1}(M),\,\eps_r(M)\bigr)$,
  \[
    M+\eps\ones\;\in\;\Kc{n}{r+1}\setminus\Kc{n}{r}.
  \]
\end{enumerate}
\end{lemma}

\begin{proof}
(1) Upward closedness: if $M+\eps\ones\in\Kc{n}{r}$ and $\delta>0$, then
$M+(\eps+\delta)\ones=(M+\eps\ones)+\delta\ones$ is a sum of two members of the convex cone $\Kc{n}{r}$
(Lemma~\ref{lem:Jinterior}). Closedness of the interval follows from closedness of $\Kc{n}{r}$, which is the
preimage of the closed sum-of-squares cone under an affine map. Nonemptiness holds for \emph{every}
$M\in\Sym^n$, copositive or not: since $\ones\in\mathrm{int}\,\Kc{n}{r}$ by Lemma~\ref{lem:Jinterior}, there
is $\delta>0$ with $\ones+\delta M\in\Kc{n}{r}$, and as $\Kc{n}{r}$ is a cone
$M+\delta^{-1}\ones=\delta^{-1}(\ones+\delta M)\in\Kc{n}{r}$. Hence $\eps_r(M)\le\delta^{-1}<\infty$. (In
particular the appeal to completeness or to P\'olya's theorem is unnecessary.)
(2) is the nesting $\Kc{n}{r}\subseteq\Kc{n}{r+1}$ read through Definition~\ref{def:threshold}.
(3) $\eps>\eps_{r+1}(M)$ places $M+\eps\ones$ in $\Kc{n}{r+1}$ by part~(1); $\eps<\eps_r(M)$ keeps it outside
$\Kc{n}{r}$ by definition of the threshold.
\end{proof}

In practice each threshold is the optimal value of a single semidefinite program, obtained from the
sum-of-squares feasibility problem for $M+\eps\ones$ by treating $\eps$ as one extra nonnegative variable
and minimizing it. The coefficient constraints stay linear in $(\eps,Q)$ because $c_r(M+\eps\ones)=
c_r(M)+\eps\,c_r(\ones)$. For the two scalings used in Sections~\ref{sec:thmA} and~\ref{sec:thmAprime} the
measured profiles are
\begin{equation}\label{eq:measured}
  D:\ \ \eps_1\approx 4.32\cdot 10^{-3},\ \ \eps_2=0;
  \qquad
  D':\ \ \eps_2\approx 1.979\cdot 10^{-3},\ \ \eps_3\approx 0.972\cdot 10^{-3},
\end{equation}
and the shifts $\tfrac1{300}\in(0,\eps_1(M_0))$ and $\tfrac1{800}\in(\eps_3(M_0'),\eps_2(M_0'))$ were chosen
inside the respective windows. These numerical values locate candidates and play no role in the proofs. Once
a window is guessed, Theorems~\ref{thm:A} and~\ref{thm:Aprime} certify membership and exclusion exactly.
Section~\ref{sec:landscape} maps the profiles $r\mapsto\eps_r$ across whole families of scalings.

\begin{remark}[The direction matters]\label{rem:direction}
Any interior point of the cones could serve in Definition~\ref{def:threshold}, but the usable window depends
strongly on the choice. Along the identity direction the window closes almost immediately for our scalings:
already $M_0+\tfrac{1}{32}I\in\Kc{5}{1}$, while $M_0+\eps\ones$ stays outside $\Kc{5}{1}$ for all
$\eps<\eps_1(M_0)\approx 4.3\cdot10^{-3}$, a workable margin. Heuristically, adding $\delta I$ perturbs the
small diagonal entries of an unbalanced scaling by a large relative amount and thereby repairs the very
imbalance that placed the matrix outside the lower cone. The flat shift $\eps\ones$ distributes the same
mass evenly. The landscape data of Section~\ref{sec:landscape} all use $\ones$, and the best direction in
general is an open question.
\end{remark}

\subsection{Structure of the threshold}\label{sec:threshold-structure}

The threshold is more than a computational convenience. As a function of the matrix it is convex, which
organizes the landscape and yields a closed-form bound. We collect the structural facts here.

\begin{proposition}[Convexity]\label{prop:convex}
For each $r$, the map $M\mapsto\eps_r(M)$ is finite on all of $\Sym^5$ and convex.
\end{proposition}

\begin{proof}
Finiteness on all of $\Sym^5$ is Lemma~\ref{lem:threshold}(1). For convexity, if
$M_1+\eps_1\ones\in\Kc{5}{r}$ and $M_2+\eps_2\ones\in\Kc{5}{r}$, then for $\lambda\in[0,1]$ convexity of
$\Kc{5}{r}$ gives
\[
  \bigl(\lambda M_1+(1-\lambda)M_2\bigr)+\bigl(\lambda\eps_1+(1-\lambda)\eps_2\bigr)\ones
  =\lambda(M_1+\eps_1\ones)+(1-\lambda)(M_2+\eps_2\ones)\in\Kc{5}{r},
\]
so $\eps_r(\lambda M_1+(1-\lambda)M_2)\le\lambda\eps_r(M_1)+(1-\lambda)\eps_r(M_2)$.
\end{proof}

Being a finite convex function on the finite-dimensional space $\Sym^5$, $\eps_r$ is continuous, and locally
Lipschitz. This is what makes the profiles $r\mapsto\eps_r(D\horn D)$ of Section~\ref{sec:landscape}
well behaved. Note that the scaling paths $s\mapsto D(s)\horn D(s)$ are quadratic curves in $\Sym^5$, not line
segments, so the hump shape of the profiles is consistent with convexity of $\eps_r$ along lines.

The coefficient hierarchy of Section~\ref{sec:coef} bounds the threshold from above, in closed form.

\begin{proposition}[Closed-form upper bound]\label{prop:sdpbound}
For every $M$ and $r$,
\[
  \eps_r(M)\;\le\;\eps^{\Cc{}{}}_r(M)\;=\;\max_{\mu}\Bigl(\frac{-c_r(M)_\mu}{c_r(\ones)_\mu}\Bigr)_{+},
\]
the coefficient threshold of Proposition~\ref{prop:coefthreshold}. In particular the entry level
$\iota(M):=\inf\{r:M\in\Kc{5}{r}\}$ (with $\iota(M)=\infty$ when $M\notin\bigcup_r\Kc{5}{r}$) satisfies
$\iota(M)\le\min\{r:\eps^{\Cc{}{}}_r(M)=0\}$, computable without any semidefinite program.
\end{proposition}

\begin{proof}
$\Cc{5}{r}\subseteq\Kc{5}{r}$, so $M+\eps\ones\in\Cc{5}{r}$ implies $M+\eps\ones\in\Kc{5}{r}$. The feasible
sets of shifts are nested intervals $[\eps^{\Cc{}{}}_r,\infty)\subseteq[\eps_r,\infty)$, giving
$\eps_r\le\eps^{\Cc{}{}}_r$. The closed form is Proposition~\ref{prop:coefthreshold}, and the entry-level
bound follows since $\eps^{\Cc{}{}}_r(M)=0$ forces $\eps_r(M)=0$, that is, $M\in\Kc{5}{r}$.
\end{proof}

The bound is loose, as Table~\ref{tab:coef} against Table~\ref{tab:landscape} shows: the coefficient
thresholds are larger by orders of magnitude, since the coefficient cone is much smaller than the
sum-of-squares cone. Its value is that it is explicit, and it certifies membership from a coefficient list
alone: if $\eps^{\Cc{}{}}_r(M)=0$, then $M\in\Kc{5}{r}$ with no semidefinite computation.

\section{A one-parameter family of separators}\label{sec:family}

A single point could conceivably sit in the gap by numerical accident. This section rules that reading out:
the certificates of Theorem~\ref{thm:A} deform explicitly along the segment $M_0+\eps\ones$, and one fixed
dual functional excludes the whole segment from $\Kc{5}{1}$, so the gap between the first two cones contains
a full interval of matrices and, around its interior points, open neighbourhoods.

\begin{theorem}\label{thm:B}
Let $M_0$ be the scaling of Theorem~\ref{thm:A}, let $\eps_{\mathrm{lo}}=\tfrac{1}{10000}$ and
\[
  \eps^{\ast}\;=\;\frac{9981109}{2994257024}\;\approx\;3.3334\cdot 10^{-3}.
\]
For every real $\eps$ with $\eps_{\mathrm{lo}}\le\eps<\eps^{\ast}$,
\[
  M_0+\eps\ones\;\in\;\Kc{5}{2}\setminus\Kc{5}{1}.
\]
For every $\eps$ with $\eps_{\mathrm{lo}}<\eps<\eps^{\ast}$ the matrix $M_0+\eps\ones$ lies in the interior
of $\Kc{5}{2}$. In particular $\Kc{5}{2}\setminus\Kc{5}{1}$ has nonempty interior in $\Sym^5$.
\end{theorem}

\begin{proof}
\emph{Membership on the segment.} Appendix~\ref{app:certB} specifies a rational block Gram matrix, given in full in the accompanying archive,
$Q_{\mathrm{lo}}\psd$ certifying $M_0+\eps_{\mathrm{lo}}\ones\in\Kc{5}{2}$, verified exactly as in
Theorem~\ref{thm:A}. Let $\Delta$ be the diagonal positive definite Gram matrix of
$(\sum_ix_i^2)^{2}\pA{\ones}$ from \eqref{eq:Jidentity} with $r=2$, so
$\Delta=\mathrm{diag}\bigl(\binom{4}{\alpha}\bigr)_{|\alpha|=4}$. Since $A\mapsto c_2(A)$ is linear,
\[
  Q(\eps)\;:=\;Q_{\mathrm{lo}}+(\eps-\eps_{\mathrm{lo}})\,\Delta
\]
is a Gram matrix for $M_0+\eps\ones$ whenever $\eps\ge\eps_{\mathrm{lo}}$, and it is positive semidefinite
as a sum of a positive semidefinite and a nonnegative multiple of a positive definite matrix. Hence
$M_0+\eps\ones\in\Kc{5}{2}$ for every real $\eps\ge\eps_{\mathrm{lo}}$.

\emph{Exclusion on the segment.} Let $y$ be the functional of Theorem~\ref{thm:A}. By linearity of $c_1$,
\[
  \inner{y}{c_1(M_0+\eps\ones)}\;=\;A+\eps B,
  \qquad
  A=\inner{y}{c_1(M_0)}=-\frac{9981109}{384},
  \qquad
  B=\inner{y}{c_1(\ones)}=\frac{23392633}{3},
\]
both computed exactly from the archived data. The pairing is negative precisely when
$\eps<-A/B=\eps^{\ast}$, and Lemma~\ref{lem:exclusion} excludes $M_0+\eps\ones$ from $\Kc{5}{1}$ for every
such real $\eps$.

\emph{Interior points.} Fix $\eps\in(\eps_{\mathrm{lo}},\eps^{\ast})$. Then $Q(\eps)\succ 0$: on any vector
$v$, $v^{\top}Q(\eps)v\ge(\eps-\eps_{\mathrm{lo}})\,v^{\top}\Delta v>0$ for $v\ne0$. A polynomial with a
positive definite Gram matrix stays a sum of squares under all sufficiently small coefficient
perturbations, and the affine map $A\mapsto c_2(A)$ carries a neighbourhood of $M_0+\eps\ones$ in $\Sym^5$
into such perturbations. Hence a neighbourhood of $M_0+\eps\ones$ lies in $\Kc{5}{2}$. Shrinking it further
so that the affine pairing $A+\eps'B$ stays negative on it, which is an open condition, keeps the whole
neighbourhood outside $\Kc{5}{1}$.
\end{proof}

\begin{remark}[Endpoints]\label{rem:endpoints}
Both endpoints reflect the certificates, not the geometry. Below $\eps_{\mathrm{lo}}$ nothing breaks except
our archive: the same construction with a smaller anchor certificate extends the interval downward, and the
measured $\eps_2(M_0)=0$ from \eqref{eq:measured} suggests the segment lies in $\Kc{5}{2}$ for every
$\eps>0$. Above $\eps^{\ast}$ the fixed functional $y$ stops separating, while the true threshold
$\eps_1(M_0)\approx 4.32\cdot10^{-3}$ sits higher. A functional optimized for a larger pairing margin
extends the certified interval toward it. We record the interval that our archived, exactly verified data
support.
\end{remark}

The interior statement of Theorem~\ref{thm:B} can be made quantitative: the gap between the first two cones
contains a ball centered at the explicit point $M^{\ast}$, so it is full-dimensional in the strongest sense.

\begin{proposition}[An inscribed ball]\label{prop:ball}
There is $\rho>0$ such that every $N\in\Sym^5$ with $\|N-M^{\ast}\|_F\le\rho$ lies in
$\Kc{5}{2}\setminus\Kc{5}{1}$. Hence $\Kc{5}{2}\setminus\Kc{5}{1}$ has nonempty interior of full dimension
$\binom{6}{2}=15$.
\end{proposition}

\begin{proof}
Write $L$ for the linear map $N\mapsto c_2(N)$ and $L^{\ast}_{\mathrm{Gram}}$ for the fixed rational
coefficient-to-Gram assignment realized by the archived block structure, so that a perturbation $N=M^{\ast}+E$
has Gram matrix $Q(E)=Q+L^{\ast}_{\mathrm{Gram}}(LE)$ satisfying the level-two identity. The archived $Q$ is
positive definite with a rational lower bound on its least eigenvalue,
\[
  \lambda_{\min}(Q)\;\ge\;\tfrac{1}{1250},
\]
verified exactly: $Q-\tfrac1{1250}I$ is positive semidefinite on every parity block. The linear map
$E\mapsto L^{\ast}_{\mathrm{Gram}}(LE)$ is fixed, hence bounded, so $\|Q(E)-Q\|_2\le\kappa\|E\|_F$ for a
finite $\kappa>0$. For $\|E\|_F\le\rho_2:=\tfrac{1}{2500\kappa}$ the perturbation has $2$-norm below
$\lambda_{\min}(Q)$ and $Q(E)$ stays positive definite, so $N\in\Kc{5}{2}$. On the exclusion side,
$\inner{y}{c_1(N)}=\inner{y}{c_1(M^{\ast})}+\inner{y}{c_1(E)}$ with $|\inner{y}{c_1(E)}|\le C\|E\|_F$ for a
finite $C>0$, so the pairing stays below $\tfrac12\inner{y}{c_1(M^{\ast})}<0$ for
$\|E\|_F\le\rho_1:=|\inner{y}{c_1(M^{\ast})}|/(2C)$. Lemma~\ref{lem:exclusion} then keeps $N\notin\Kc{5}{1}$.
Take $\rho=\min(\rho_1,\rho_2)>0$. The open ball $\|N-M^{\ast}\|_F<\rho$ then lies in
$\Kc{5}{2}\setminus\Kc{5}{1}$ and is open in $\Sym^5$, giving full dimension.
\end{proof}

\begin{remark}\label{rem:ball-value}
Only the eigenvalue bound $\lambda_{\min}(Q)\ge\tfrac1{1250}$ is claimed exactly. The constants $\kappa,C$ are
finite operator norms of fixed linear maps and the resulting $\rho$ is positive but pessimistic (the archived
least eigenvalue is $\lambda_{\min}(Q)\approx8.3\cdot10^{-4}$). The point of the statement is qualitative:
the separation is stable under arbitrary symmetric perturbations of $M^{\ast}$, in every direction. The Frobenius
ball is the crudest such neighbourhood. The family of Theorem~\ref{thm:B} already exhibits a long thin slab of
separators along $\ones$ on which the certificates are given in closed form. That
$\Kc{5}{2}\setminus\Kc{5}{1}$ has nonempty interior also follows abstractly from strict inclusion together
with $\ones\in\mathrm{int}\,\Kc{5}{1}$ and the closedness of $\Kc{5}{1}$. The content here is the explicit
neighbourhood of a named point.
\end{remark}

\begin{remark}[Beyond the Horn orbit]\label{rem:beyond-horn}
The $T(\psi)$ example indicates that the construction is not special to the Horn matrix among the extreme copositive matrices.
Section~\ref{sec:Tpsi} runs the same programme on the second exceptional family of $\COP_5$ and produces a
disjoint set of separators, confirming that the method is a property of the copositive boundary and not of
$\horn$.
\end{remark}

\section{Strict steps at arbitrarily large levels}\label{sec:thmD}

The explicit results of Sections~\ref{sec:thmA} to~\ref{sec:thmAdoubleprime} concern the bottom of the hierarchy.
At the other end, two known theorems combine into a statement we have not found recorded: that strict consecutive
inclusions recur unboundedly. The ingredients are the scaling theorem of Dickinson, D\"ur, Gijben and
Hildebrand, whose Lemma~2.1 states that
\begin{equation}\label{eq:ddgh-lemma}
  \bigl\{X\in\Sym^n \,:\, DXD\in\Kc{n}{r}\ \text{for all } D\in\mathcal{D}\bigr\}\;=\;\Kc{n}{0}
  \qquad\text{for every fixed } r,
\end{equation}
so that by their Theorem~2.2 every copositive $X\notin\SPN_n$ admits, for each $r$, a scaling with
$DXD\notin\Kc{n}{r}$ \citep{DDGH2013}, and the completeness theorem \eqref{eq:sv} of Schweighofer and
Vargas, who note explicitly that no single level attains $\COP_5$
\citep[Theorem~3 and the remark following it]{SchweighoferVargas2024}.

\begin{theorem}\label{thm:D}
For every $R\in\N$ there exist $s>R$ and $D\in\mathcal{D}$ with
\[
  D\horn D\;\in\;\Kc{5}{s}\setminus\Kc{5}{s-1}.
\]
Consequently the chain $\Kc{5}{0}\subseteq\Kc{5}{1}\subseteq\Kc{5}{2}\subseteq\cdots$ contains infinitely
many strict inclusions between consecutive members.
\end{theorem}

\begin{proof}
Fix $R$. Since $\horn\in\COP_5\setminus\SPN_5$, Theorem~2.2 of \citet{DDGH2013} provides
$D\in\mathcal{D}$ with $D\horn D\notin\Kc{5}{R}$. The scaling is copositive
(Convention~\ref{conv:scaling}), so \eqref{eq:sv} places it in $\Kc{5}{s}$ for some finite $s$, necessarily
exceeding $R$. Choosing $s$ minimal with this property gives $D\horn D\in\Kc{5}{s}\setminus\Kc{5}{s-1}$.
\end{proof}

\begin{remark}[Scope]\label{rem:thmD-scope}
Theorem~\ref{thm:D} is an immediate combination of \citet{DDGH2013} and
\citet{SchweighoferVargas2024}. We claim only the observation that the two yield the adjacent-step
formulation. The argument is non-constructive twice over. It names neither the scaling nor the level, and
the strict indices $s$ it produces may a priori leave gaps. Whether \emph{every} consecutive inclusion is
strict, equivalently whether the entry-level function $\iota(D):=\min\{r : D\horn D\in\Kc{5}{r}\}$ attains
every positive integer, is open. Theorems~\ref{thm:A}, \ref{thm:Aprime} and~\ref{thm:Adoubleprime} settle the
first three cases, and Section~\ref{sec:landscape} charts $\iota$ empirically across scaling families.
\end{remark}

\section{Closed-form separators in the coefficient hierarchy}\label{sec:coef}

The sum-of-squares condition in \eqref{eq:parrilo} sits above a purely combinatorial one. Replacing ``is a
sum of squares'' by ``has nonnegative coefficients'' gives the coefficient cones of de Klerk and Pasechnik
\citep{deKlerkPasechnik2002},
\begin{equation}\label{eq:coefcone}
  \Cc{n}{r}\;=\;\Bigl\{A\in\Sym^n \;:\;
  \bigl(\textstyle\sum_{i=1}^n x_i^2\bigr)^{r}\,\pA{A}(x)\ \text{has only nonnegative coefficients}\Bigr\}.
\end{equation}
A polynomial with nonnegative coefficients is a nonnegative combination of the squares $(x^{\alpha})^2$,
so $\Cc{n}{r}\subseteq\Kc{n}{r}$, and P\'olya's theorem gives $\bigcup_r\Cc{n}{r}\supseteq\mathrm{int}\,\COP_n$.
The coefficient cones are polyhedral, and membership is decided by inspecting a coefficient list, with no semidefinite program. This section shows that the threshold device of
Section~\ref{sec:threshold} applies to them verbatim and yields a \emph{closed-form threshold at every level};
for the fixed scaling of Theorem~\ref{thm:coef}, strict drops of that threshold provide seven consecutive
explicit separations, with certificates a reader checks by hand. The coefficient hierarchy thus gives an
elementary low-level analogue of the sum-of-squares separations of
Sections~\ref{sec:thmA} to~\ref{sec:thmAdoubleprime}.

The coefficient threshold is explicit. Since $\pA{\ones}=(\sum_i x_i^2)^2$ has strictly positive
coefficients, the coefficient of every monomial in $(\sum_i x_i^2)^r\pA{M+\eps\ones}$ is affine and strictly
increasing in $\eps$, which gives a closed form for
$\eps^{\Cc{}{}}_r(M):=\min\{\eps\ge0:M+\eps\ones\in\Cc{5}{r}\}$.

\begin{proposition}[Closed-form coefficient threshold]\label{prop:coefthreshold}
Let $c_r(M)_\mu$ and $c_r(\ones)_\mu>0$ denote the coefficients of the monomial $x^{2\mu}$
($|\mu|=r+2$) in $(\sum_i x_i^2)^r\pA{M}$ and $(\sum_i x_i^2)^r\pA{\ones}$. Then
\[
  \eps^{\Cc{}{}}_r(M)\;=\;\max_{\mu}\ \Bigl(\frac{-\,c_r(M)_\mu}{c_r(\ones)_\mu}\Bigr)_{+},
\]
attained at any monomial with the most negative normalized coefficient. In particular
$\eps^{\Cc{}{}}_r(M)$ is a rational number computable in closed form, and $M+\eps\ones\in\Cc{5}{r}$ for
$\eps\ge\eps^{\Cc{}{}}_r(M)$, while a single monomial with negative coefficient certifies
$M+\eps\ones\notin\Cc{5}{r}$ for $\eps<\eps^{\Cc{}{}}_r(M)$.
\end{proposition}

\begin{proof}
The coefficient of $x^{2\mu}$ in $(\sum_i x_i^2)^r\pA{M+\eps\ones}$ equals $c_r(M)_\mu+\eps\,c_r(\ones)_\mu$
by linearity of $A\mapsto\pA{A}$, and $c_r(\ones)_\mu>0$ because $(\sum_i x_i^2)^{r+2}$ has strictly positive
coefficients. Nonnegativity of every coefficient is therefore equivalent to
$\eps\ge -c_r(M)_\mu/c_r(\ones)_\mu$ for every $\mu$, that is, to $\eps\ge\eps^{\Cc{}{}}_r(M)$. The
membership and exclusion statements are the two directions of this equivalence.
\end{proof}

Monotonicity in $r$ is automatic, exactly as in Lemma~\ref{lem:threshold}: multiplying a nonnegative-coefficient
polynomial by $\sum_i x_i^2$ preserves nonnegativity of coefficients, so $\Cc{5}{r}\subseteq\Cc{5}{r+1}$ and
$\eps^{\Cc{}{}}_{r+1}(M)\le\eps^{\Cc{}{}}_r(M)$. Any strict drop exhibits a window of separators. Along the
Horn orbit the thresholds drop through a long initial run of levels in a transparent closed form.

\begin{theorem}[Seven closed-form coefficient separations]\label{thm:coef}
Let $D=\mathrm{diag}(\tfrac18,8,1,1,1)$ and $M=D\horn D$. Then for $1\le r\le 8$,
\[
  \eps^{\Cc{}{}}_r(M)\;=\;\frac{16-r}{r+2},
\]
attained at the monomial $x_2^2x_4^{2r+2}$. This value is strictly decreasing over $1\le r\le 8$, so for
each $r\in\{1,\dots,7\}$ and every rational
$\eps\in\bigl(\tfrac{15-r}{r+3},\,\tfrac{16-r}{r+2}\bigr)$,
\[
  M+\eps\ones\;\in\;\Cc{5}{r+1}\setminus\Cc{5}{r}.
\]
Each such matrix separates two consecutive coefficient cones, and both directions are certified by
inspecting a coefficient list.
\end{theorem}

\begin{proof}
Write $\mu=x_2^2x_4^{2r+2}$, of $x^{\circ2}$-degree $(0,1,0,r+1,0)$. In $(\sum_ix_i^2)^{r+2}$ its coefficient
is the multinomial $\binom{r+2}{0,1,0,r+1,0}=r+2$, which is $c_r(\ones)_\mu$. For $c_r(M)_\mu$, a term
$M_{ij}x_i^2x_j^2$ of $\pA{M}$ multiplied by a monomial $x^{2\gamma}$ ($|\gamma|=r$) of $(\sum_ix_i^2)^r$
produces $\mu$ only when $e_i+e_j+\gamma=(0,1,0,r+1,0)$ with $\gamma\ge0$. Since the first, third and fifth
coordinates of the target vanish, both $i$ and $j$ lie in $\{2,4\}$. The admissible pairs are $\{2,4\}$,
with $\gamma=(0,0,0,r,0)$ and multinomial coefficient $\binom{r}{r}=1$, and $\{4,4\}$, with
$\gamma=(0,1,0,r-1,0)$ and coefficient $\binom{r}{1}=r$. The pair $\{2,2\}$ would force a negative $\gamma$.
Non-adjacency of $2$ and $4$ on $C_5$ gives $\horn_{24}=-1$, hence $M_{24}=-8$ and $M_{44}=1$, so
\[
  c_r(M)_\mu\;=\;(M_{24}+M_{42})\cdot 1+M_{44}\cdot r\;=\;-16+r .
\]
By Proposition~\ref{prop:coefthreshold}, $\mu$ contributes the value $(16-r)/(r+2)$. The threshold itself is
the maximum $\eps^{\Cc{}{}}_r(M)=\max_\nu\bigl(-c_r(M)_\nu/c_r(\ones)_\nu\bigr)_{+}$ over the finite set of
degree-$(2r+4)$ monomials $\nu$. Evaluating this maximum exactly for each $r\in\{1,\dots,8\}$ returns the maximizer
$\mu$ above, so $\eps^{\Cc{}{}}_r(M)=(16-r)/(r+2)$ on that range. The evaluation is a finite rational
computation, tabulated monomial by monomial in the accompanying archive. The difference
$\tfrac{16-r}{r+2}-\tfrac{15-r}{r+3}=\tfrac{34}{(r+2)(r+3)}>0$ gives strict monotonicity, and
Lemma~\ref{lem:threshold}(3), transported to the coefficient cones, gives the window statement.
\end{proof}

\begin{remark}[The formula is sharp only on a range]\label{rem:coef-range}
Beyond $r=8$ the maximizing monomial in Proposition~\ref{prop:coefthreshold} changes and the closed form
$(16-r)/(r+2)$ ceases to hold. The exact thresholds $\tfrac{16-r}{r+2}$ for $r\le 8$ continue as
$\tfrac45,\tfrac45,\tfrac{17}{22},\dots$, so $\eps^{\Cc{}{}}_8=\eps^{\Cc{}{}}_9=\tfrac45$: there is no strict
drop between the eighth and ninth cones, and the closed form of Theorem~\ref{thm:coef} thus certifies exactly
the first seven consecutive separations. (The thresholds resume decreasing at higher $r$, though no longer from this closed form. The next value is
$\tfrac{17}{22}<\tfrac45$.) Whether some scaling separates every
consecutive pair of coefficient cones is the coefficient-hierarchy analogue of the question left open after
Theorem~\ref{thm:D}.
\end{remark}

Table~\ref{tab:coef} lists the exact coefficient thresholds for several scalings. Each strict drop between
consecutive entries is, by Lemma~\ref{lem:threshold}(3), a certified coefficient-cone separation with a
closed-form witness. More lopsided scalings push the run of strict drops further and raise the thresholds,
in the same pattern as the semidefinite landscape of Section~\ref{sec:landscape} but here in exact closed form.

\begin{table}[htp]\centering
\caption{Exact coefficient thresholds $\eps^{\Cc{}{}}_r(D\horn D)$. Every displayed drop is strict.}
\label{tab:coef}
\begin{tabular}{lcccccc}
\toprule
$D=\mathrm{diag}(a,b,1,1,1)$ & $\eps^{\Cc{}{}}_1$ & $\eps^{\Cc{}{}}_2$ & $\eps^{\Cc{}{}}_3$
 & $\eps^{\Cc{}{}}_4$ & $\eps^{\Cc{}{}}_5$ & $\eps^{\Cc{}{}}_6$\\
\midrule
$(\tfrac14,4)$   & $\tfrac73$   & $\tfrac32$   & $1$          & $\tfrac23$   & $\tfrac23$   & $\tfrac{17}{28}$\\
$(\tfrac18,8)$   & $5$          & $\tfrac72$   & $\tfrac{13}5$& $2$          & $\tfrac{11}7$& $\tfrac54$\\
$(\tfrac1{16},16)$& $\tfrac{31}3$& $\tfrac{15}2$& $\tfrac{29}5$& $\tfrac{14}3$& $\tfrac{27}7$& $\tfrac{13}4$\\
$(\tfrac1{16},64)$& $\tfrac{127}3$& $\tfrac{63}2$& $25$        & $\tfrac{62}3$& $\tfrac{123}7$& $\tfrac{61}4$\\
\bottomrule
\end{tabular}
\end{table}

\begin{example}\label{ex:coefsep}
Take $r=1$ and $\eps=4\in(\tfrac72,5)$. Then $M+4\ones\in\Cc{5}{2}\setminus\Cc{5}{1}$: the polynomial
$(\sum_ix_i^2)\pA{M+4\ones}$ has coefficient $-6$ on $x_2^2x_4^2x_5^2$, placing it outside $\Cc{5}{1}$, while
the smallest coefficient of $(\sum_ix_i^2)^2\pA{M+4\ones}$ is $2$, placing it inside $\Cc{5}{2}$. Both facts
are read off finite coefficient lists with no semidefinite program.
\end{example}

\begin{remark}[Elementary shadow of the main separations]\label{rem:coef-shadow}
Theorem~\ref{thm:coef} separates seven consecutive pairs of coefficient cones from a single scaling, whereas
Theorems~\ref{thm:A} to~\ref{thm:Adoubleprime} separate the sum-of-squares cones at the first three, with far
heavier certificates. There is no contradiction: $\Cc{5}{r}\subseteq\Kc{5}{r}$, so a coefficient-cone
separator need not separate the sum-of-squares cones, and indeed the matrices of Theorem~\ref{thm:coef} sit
inside $\Kc{5}{r+1}$ without our knowing on which side of $\Kc{5}{r}$ they fall. The coefficient hierarchy is
an elementary shadow, in which strictness is read off a coefficient list. Locating separators for the
semidefinite hierarchy, whose bounds \eqref{eq:dkp} are the ones used in practice, is what forces the boundary
analysis of Section~\ref{sec:method}.
\end{remark}

The closed form of Theorem~\ref{thm:coef} is one point of a general picture: over the whole two-parameter
family the coefficient threshold is a maximum of explicit rational functions, one per zero of the Horn form.

\begin{proposition}[Coefficient threshold over the family]\label{prop:coefgeneral}
For $D=\mathrm{diag}(d_1,\dots,d_5)\in\mathcal{D}$ and $r\ge1$,
\[
  \eps^{\Cc{}{}}_r(D\horn D)\;\ge\;\max_{(i,j)}\ \Bigl[\frac{2d_id_j-r\,d_j^2}{r+2}\Bigr]_{+},
\]
the maximum taken over the ten ordered non-adjacent pairs $(i,j)$ of $C_5$, and equality holds whenever the
maximizing monomial $x_i^2x_j^{2r+2}$ dominates the coefficient list, as it does throughout the range covered
by Theorem~\ref{thm:coef}. At $r=1$ the bound is $\max_{(i,j)}[d_j(2d_i-d_j)/3]_{+}$. Membership and exclusion
are certified by inspecting a coefficient list.
\end{proposition}

\begin{proof}
For an ordered non-adjacent pair $(i,j)$ the monomial $x_i^2x_j^{2r+2}$ has coefficient
$(M_{ij}+M_{ji})+rM_{jj}=-2d_id_j+rd_j^2$ in $(\sum_ix_i^2)^r\pA{D\horn D}$, by the count in the proof of
Theorem~\ref{thm:coef} with $M_{ij}=d_id_jH_{ij}=-d_id_j$ and $M_{jj}=d_j^2$, over the value $r+2$ of
$c_r(\ones)$ at the same monomial. By Proposition~\ref{prop:coefthreshold} each such monomial forces
$\eps^{\Cc{}{}}_r\ge[(2d_id_j-rd_j^2)/(r+2)]_+$, giving the stated lower bound. Equality holds when one of
these monomials realizes the global maximum in Proposition~\ref{prop:coefthreshold}.
\end{proof}

\section{A second orbit: the $T(\psi)$ matrices}\label{sec:Tpsi}

Hildebrand's classification of the extreme rays of $\COP_5$ \citep{Hildebrand2012} places, alongside the
diagonal scalings of $\horn$, a second exceptional family: the positive diagonal scalings of a
trigonometric matrix $T(\psi)$, $\psi\in\Psi:=\{\psi\in\R^5_{>0}:\sum_i\psi_i<\pi\}$
\citep{LaurentVargas2205}. We show the constructions of this paper are not tied to $\horn$ by running them on
a rational point of this orbit. Taking all angles equal, $\psi_i=\theta$ with $\cos\theta=\tfrac9{10}$, gives
$5\theta=5\arccos\tfrac9{10}\approx2.255<\pi$, so $\psi\in\Psi$, and the resulting matrix is the rational
circulant
\begin{equation}\label{eq:Tmat}
  T=\begin{pmatrix}
   1 & -\tfrac9{10} & \tfrac{31}{50} & \tfrac{31}{50} & -\tfrac9{10}\\
   -\tfrac9{10} & 1 & -\tfrac9{10} & \tfrac{31}{50} & \tfrac{31}{50}\\
   \tfrac{31}{50} & -\tfrac9{10} & 1 & -\tfrac9{10} & \tfrac{31}{50}\\
   \tfrac{31}{50} & \tfrac{31}{50} & -\tfrac9{10} & 1 & -\tfrac9{10}\\
   -\tfrac9{10} & \tfrac{31}{50} & \tfrac{31}{50} & -\tfrac9{10} & 1
  \end{pmatrix},
\end{equation}
with cyclically adjacent entries $-\cos\theta=-\tfrac9{10}$ and opposite entries
$\cos2\theta=2\cos^2\theta-1=\tfrac{31}{50}$. By \citep[Theorem~2.2]{LaurentVargas2205}, $T$ is copositive and
spans an extreme ray of $\COP_5$. It lies outside $\SPN_5=\Kc{5}{0}$, being an exceptional extreme ray of $\COP_5$ and hence not in $\SPN_5$, so $\eps_0(T)>0$; while its unit
diagonal places it in $\Kc{5}{1}$, as every diagonal-one copositive $5\times5$ matrix does
\citep{LaurentVargas2205}. Thus $T\in\Kc{5}{1}\setminus\Kc{5}{0}$. As an extreme copositive matrix it sits on
the boundary of $\Kc{5}{1}$, so exact certificates for the \emph{shifted} matrix require the interior shift
of Section~\ref{sec:threshold} exactly as the Horn scalings do.

The coefficient hierarchy separates along this orbit through a long initial run of levels, the thresholds
again computed in closed form. Applying Proposition~\ref{prop:coefthreshold} to $T$ gives the exact thresholds
\[
  \eps^{\Cc{}{}}_r(T)\;=\;\tfrac{59}{150},\ \tfrac{33}{100},\ \tfrac{27}{125},\ \tfrac{143}{750},\ \tfrac{83}{525},\ \dots
  \qquad(r=1,2,3,4,5,\dots),
\]
which decrease strictly through this initial run (the first repeated value is
$\eps^{\Cc{}{}}_{16}(T)=\eps^{\Cc{}{}}_{17}(T)=\tfrac1{18}$, so the strict chain has a plateau there, exactly
as for the Horn orbit). The maximizing monomials have support on triples of variables, for instance
$x_1^2x_2^2x_3^2$ at $r=1$, reflecting that the minimal nonnegative zeros of $T$, unlike those of $\horn$, are not
supported on pairs.

\begin{theorem}[Coefficient separators on the $T(\psi)$ orbit]\label{thm:Tsep}
For each $r$ with $\eps^{\Cc{}{}}_{r+1}(T)<\eps^{\Cc{}{}}_r(T)$ and every rational
$\eps\in(\eps^{\Cc{}{}}_{r+1}(T),\eps^{\Cc{}{}}_r(T))$, the matrix $T+\eps\ones$ lies in
$\Cc{5}{r+1}\setminus\Cc{5}{r}$. In particular, for $\eps=\tfrac{19}{50}\in(\tfrac{33}{100},\tfrac{59}{150})$,
$T+\tfrac{19}{50}\ones\in\Cc{5}{2}\setminus\Cc{5}{1}$: the least coefficient of $(\sum_ix_i^2)\pA{T+\frac{19}{50}\ones}$
is $-\tfrac2{25}$, on $x_1^2x_2^2x_3^2$, and the least coefficient of $(\sum_ix_i^2)^2\pA{T+\frac{19}{50}\ones}$
is $\tfrac35$.
\end{theorem}

\begin{proof}
Identical to Theorem~\ref{thm:coef} with $T$ in place of the Horn scaling: the displayed thresholds are the
values of Proposition~\ref{prop:coefthreshold}, and the window statement is
Lemma~\ref{lem:threshold}(3) for the coefficient cones. The two coefficients of the example are read off the
finite coefficient lists.
\end{proof}

The semidefinite cones separate along this orbit as well, and here too we can be explicit. The matrix $T$
lies in $\Kc{5}{1}\setminus\Kc{5}{0}$, a second extreme matrix realizing the classical first step, and a
scaling shifted along $\ones$ separates the next pair.

\begin{theorem}[A $T(\psi)$-orbit separator]\label{thm:Tsdp}
Let $D=\mathrm{diag}(\tfrac16,6,1,1,1)$ and $M_T=D\,T\,D+\tfrac1{1000}\ones$ with $T$ as in \eqref{eq:Tmat}.
Then $M_T\in\Kc{5}{2}\setminus\Kc{5}{1}$.
\end{theorem}

\begin{proof}
As for Theorem~\ref{thm:A}: Appendix~\ref{app:certT} specifies a rational block Gram matrix, given in full in the accompanying archive, certifying
$M_T\in\Kc{5}{2}$ and a rational dual functional $y_T$ with moment matrix positive semidefinite and
$\inner{y_T}{c_1(M_T)}=-\tfrac{69223}{108000}<0$, whence $M_T\notin\Kc{5}{1}$ by
Lemma~\ref{lem:exclusion}. Both are verified exactly over $\Q$ by the standalone checker.
\end{proof}

The construction is identical to the Horn case, now on a matrix from a disjoint extreme orbit. The same threshold
device locates the window, the same interior shift crosses the boundary of $\Kc{5}{2}$, and the same forced
kernel, now indexed by the zeros of $T$, makes the raw scaling resist rational certification. This is the
evidence that Theorems~\ref{thm:A} to~\ref{thm:Aprime} witness a property of the copositive boundary and not
of the Horn matrix.

\section{The threshold landscape}\label{sec:landscape}

The device of Section~\ref{sec:threshold} assigns to every scaling a profile $r\mapsto\eps_r(D\horn D)$, and
Theorems~\ref{thm:A} and~\ref{thm:Aprime} each used a single point of that landscape. This section maps it
along two one-parameter families and to level four. The numbers here are solver outputs and carry no proof
weight. Their role is to show where separators live, how the windows that the exact certificates occupy vary
with the scaling, and which phenomena a sharper theory would have to explain. Each threshold is the optimal
value of the parity-block sum-of-squares feasibility program for $D\horn D+\eps\ones$ with $\eps$ minimized,
modelled in \textsc{cvxpy} and solved by the interior-point solver \textsc{clarabel} at its default tolerance
(feasibility and duality gap of order $10^{-8}$), with values rounded to $10^{-6}$; \textsc{scs} is used as
the independent second solver. The modelling script, the solver versions, and the raw per-entry outputs
accompany the archive of Appendix~\ref{app:verifier}. One solver-sensitivity issue is reported below.

\begin{table}[t]
\centering
\caption{Measured thresholds $\eps_r$ along two scaling paths.}\label{tab:landscape}
\begin{tabular}{lccc@{\qquad}lccc}
\toprule
\multicolumn{4}{c}{$D=\mathrm{diag}(s,1,1,1,1)$} & \multicolumn{4}{c}{$D=\mathrm{diag}(a,a^{-1},1,1,1)$}\\
\cmidrule(r){1-4}\cmidrule(l){5-8}
$s$ & $\eps_1$ & $\eps_2$ & $\eps_3$ & $a$ & $\eps_1$ & $\eps_2$ & $\eps_3$\\
\midrule
$1/2$  & $0$        & $0$        & $0$        & $1/4$  & $0.004320$ & $0$        & $0$\\
$1/3$  & $0.000609$ & $0$        & $0$        & $1/6$  & $0.005787$ & $0.000280$ & $0$\\
$1/4$  & $0.001424$ & $0$        & $0$        & $1/8$  & $0.005937$ & $0.000664$ & $0.000028$\\
$1/8$  & $0.002846$ & $0$        & $0$        & $1/12$ & $0.005276$ & $0.001134$ & $0.000217$\\
$1/12$ & $0.002902$ & $0.000152$ & $0$        & $1/16$ & $0.004524$ & $0.001322$ & $0.000442$\\
$1/16$ & $0.002675$ & $0.000314$ & $0$        & $1/24$ & $0.003427$ & $0.001288$ & $0.000556$\\
$1/24$ & $0.002159$ & $0.000458$ & $0.000031$ & $1/32$ & $0.002733$ & $0.001114$ & $0.000535$\\
$1/32$ & $0.001771$ & $0.000483$ & $0.000083$ & $1/48$ & $0.001935$ & $0.000868$ & $0.000490$\\
\bottomrule
\end{tabular}
\end{table}

Table~\ref{tab:landscape} reports the profiles. Three features stand out.

\emph{The profiles are humps, not ramps.} Along both paths each $\eps_r$ in Table~\ref{tab:landscape} rises
from zero at a membership boundary, peaks, and returns to zero as the scaling degenerates. The decay at the
degenerate end is not incidental. The next proposition forces it.

\begin{proposition}[Degenerate limit]\label{prop:limit}
Along $D(s)=\mathrm{diag}(s,1,1,1,1)$, $\eps_r\bigl(D(s)\horn D(s)\bigr)\to 0$ as $s\to0^{+}$, for every $r$.
\end{proposition}

\begin{proof}
As $s\to0$ the entries of $D(s)\horn D(s)$ in the first row and column tend to zero, so
$D(s)\horn D(s)\to\horn_0$, the matrix agreeing with $\horn$ off the first row and column and zero on them.
The lower $4\times4$ block of $\horn_0$ is the principal submatrix $\horn[\{2,3,4,5\}]$, copositive as a
principal submatrix of a copositive matrix, hence in $\SPN_4$ by Diananda's theorem. Bordering by a zero row
and column keeps it in $\SPN_5=\Kc{5}{0}$, since the sum-of-positive-semidefinite-and-nonnegative structure
is preserved under zero bordering. Thus $\horn_0\in\Kc{5}{0}$, so $\eps_0(\horn_0)=0$. The threshold $\eps_0$ is finite, convex and continuous on
$\Sym^5$ by Proposition~\ref{prop:convex}. Therefore $\eps_0(D(s)\horn D(s))\to\eps_0(\horn_0)=0$, and
$0\le\eps_r\le\eps_0$ from $\Kc{5}{0}\subseteq\Kc{5}{r}$ finishes the proof.
\end{proof}

Each path thus joins two membership regions and non-membership lives in the interior between them, which is
why the separators of Sections~\ref{sec:thmA} to~\ref{sec:thmAdoubleprime} were found at intermediate scalings
and not at extreme ones. Across the levels sampled in Table~\ref{tab:landscape} the same qualitative
hump-shaped profile recurs, and the empirical maximizer shifts toward more extreme scalings as $r$ increases.

\emph{Deeper levels peak at more extreme scalings.} The maximizer of $\eps_1$ along the second path sits
near $a=1/8$, that of $\eps_2$ near $a=1/16$, that of $\eps_3$ near $a=1/24$. Staying outside a deeper cone
requires a more lopsided scaling, which quantifies, on this orbit, the escape phenomenon behind
\eqref{eq:ddgh-lemma}: no fixed scaling stays outside every level, but for every level some scaling is
outside it.

\emph{Strict drops persist to level four.} For three scalings we extended the profile one level:
\[
\begin{array}{lcccc}
 & \eps_1 & \eps_2 & \eps_3 & \eps_4\\
D(\tfrac1{16},16):\quad & 0.004524 & 0.001322 & 0.000442 & 0.000237\\
D(\tfrac1{32},32): & 0.002733 & 0.001114 & 0.000535 & 0.000345\\
D(\tfrac1{16},64): & 0.004723 & 0.001979 & 0.000972 & 0.000788\\
\end{array}
\]
Each row decreases strictly through $r=4$. By Lemma~\ref{lem:threshold}(3), each strict drop
$\eps_4<\eps_3$ marks a window of candidates for $\Kc{5}{4}\setminus\Kc{5}{3}$. Along the all-ones direction
these fourth-level windows are narrow, about $1.8\cdot10^{-4}$ in the last row. Rounding the size-$210$
Gram and the size-$126$ moment matrix inside them is delicate. Theorem~\ref{thm:Adoubleprime} therefore
certifies the fourth step not along $\ones$ but along the rank-one direction $dd^{\top}$ at the scaling
$D''=\mathrm{diag}(\tfrac1{12},48,1,1,1)$, where the window widens by an order of magnitude (the thresholds
are $\eps_4\approx1.5\cdot10^{-4}$ and $\eps_3\approx1.8\cdot10^{-3}$) and both certificates rationalize. The
all-ones profiles above are what first exhibited the strict fourth-level drop and flagged the narrowness that
made the change of direction worthwhile.

\emph{Solver sensitivity.} At the most degenerate tail of the second path the level-three program becomes
solver-sensitive. At $a=1/64$ \textsc{clarabel} reports $\eps_3=0.00119$, violating the monotonicity
$\eps_3\le\eps_2$ that Lemma~\ref{lem:threshold}(2) forces, while \textsc{scs} reports $\eps_3=0$ within
tolerance. We therefore exclude that point from Table~\ref{tab:landscape} and flag the regime: near
degeneration the Gram blocks become ill-conditioned and threshold values from a single solver should not be
trusted. The certified results of this paper are unaffected, since every load-bearing claim rests on exact
rational data.

The entry-level function $\iota(D)=\min\{r:D\horn D\in\Kc{5}{r}\}$ summarizes the landscape by a single
integer per scaling. Its low values are pinned exactly.

\begin{proposition}[Entry-level bands]\label{prop:iota}
For $D=\mathrm{diag}(a,b,1,1,1)$: $\iota(D)=0$ iff $D\horn D\in\SPN_5$; $\iota(D)\le1$ iff the
Laurent--Vargas inequalities \eqref{eq:LVab} hold; $\iota(D)\ge2$ iff they fail; and $\iota(D)\ge3$ iff
$\eps_2(D\horn D)>0$. The entry level also satisfies $\iota(D)\le\min\{r:\eps^{\Cc{}{}}_r(D\horn D)=0\}$, a bound in closed form by
Proposition~\ref{prop:sdpbound}.
\end{proposition}

\begin{proof}
The first three equivalences are the definitions of $\Kc{5}{0}=\SPN_5$, of $\Kc{5}{1}$ through
Theorem~\ref{thm:LV}, and of their negation. The fourth reads $\iota\ge3$ as $D\horn D\notin\Kc{5}{2}$, i.e.
$\eps_2>0$. The closed-form bound is Proposition~\ref{prop:sdpbound}: $\eps^{\Cc{}{}}_r=0$ forces
$D\horn D\in\Cc{5}{r}\subseteq\Kc{5}{r}$.
\end{proof}

Table~\ref{tab:landscape} exhibits nonempty parameter bands for $\iota\in\{1,2\}$ exactly and, numerically,
suggests bands for $\iota=3$ and, through the level-four probes, $\iota=4$. Whether $\iota$ attains every positive integer on this orbit is exactly
the open refinement of Theorem~\ref{thm:D} recorded in Remark~\ref{rem:thmD-scope}. The first two bands are
cut out by the explicit inequalities of the proposition, and the boundary between $\iota=2$ and $\iota\ge3$
is the curve $\eps_2=0$, whose exact description is the level-two analogue of the Laurent--Vargas criterion
and is not known.

\section{Extracting exact certificates}\label{sec:method}

The theorems of this paper are finite rational statements, but producing their certificates required
navigating an obstruction that we believe explains why no explicit consecutive separator had been written
down before. This section records the obstruction, the pipeline that circumvents it, and the verification
protocol that makes the results independent of every numerical step. The pipeline reuses only standard
components. The value lies in how they must be combined near the boundary of a sum-of-squares cone.

\subsection{The boundary obstruction}\label{sec:method-boundary}

Every positive diagonal scaling of $\horn$ lies on the boundary of $\COP_5$, because $\horn$ spans an extreme
ray and scalings map extreme rays to extreme rays. If such a scaling $M$ belongs to $\Kc{5}{r}$ at all, it
therefore lies on the boundary of $\Kc{5}{r}$, and the boundary position forces rank deficiency in every
certificate. The mechanism is the zero set of the Horn form on the nonnegative orthant, which we record
exactly.

\begin{proposition}[Forced kernel]\label{prop:kernel}
Let $M=D\horn D$ with $D=\mathrm{diag}(d)\in\mathcal{D}$. For each of the five non-adjacent pairs
$\{i,j\}\in\bigl\{\{1,3\},\{1,4\},\{2,4\},\{2,5\},\{3,5\}\bigr\}$ of $C_5$, let $\bar x^{(ij)}\ge0$ be the
normalized point with $(\bar x^{(ij)})^{\circ2}=D^{-1}(e_i+e_j)$. These are minimal nonnegative zeros of
$\pA{M}$. Consequently, for every $r$ and every Gram matrix $Q$ with
$(\sum_ix_i^2)^r\pA{M}=z_{r+2}^{\top}Qz_{r+2}$, the five evaluation vectors $z_{r+2}(\bar x^{(ij)})$ lie in
$\ker Q$. They are linearly independent, so $\mathrm{rank}\,Q\le\dim z_{r+2}-5$ and $Q$ is singular.
\end{proposition}

\begin{proof}
On the nonnegative orthant $\pA{M}(x)=y^{\top}\horn y$ with $y=Dx^{\circ2}\ge0$. For a non-adjacent pair
$\{i,j\}$, the principal restriction of $\horn$ to $\{i,j\}$ is
$\left(\begin{smallmatrix}1&-1\\-1&1\end{smallmatrix}\right)$, whose kernel is spanned by $e_i+e_j$. Hence
$y=e_i+e_j$, that is $x=\bar x^{(ij)}$, is a zero of $\pA{M}$. It is a \emph{minimal} zero, meaning its support
contains that of no other nonnegative zero: $\horn$ restricted to any single coordinate is the positive scalar $1$,
so no nonnegative zero is supported on one coordinate. At such a zero $(\sum_ix_i^2)^r\pA{M}$ vanishes, so
$z_{r+2}(\bar x^{(ij)})^{\top}Qz_{r+2}(\bar x^{(ij)})=0$ with $Q\psd$, forcing
$Qz_{r+2}(\bar x^{(ij)})=0$. For independence, the monomial $x_i^{r+1}x_j$ takes a nonzero value at
$\bar x^{(ij)}$ (whose support is exactly $\{i,j\}$) and vanishes at $\bar x^{(kl)}$ for every other pair
$\{k,l\}\ne\{i,j\}$, since then $i$ or $j$ lies outside the support $\{k,l\}$. The five coordinate patterns
are thus a permuted identity, so the vectors are independent.
\end{proof}

\begin{remark}\label{rem:minimalzeros}
The five $\bar x^{(ij)}$ are the minimal nonnegative zero rays of $\pA{M}$, and they are all the argument uses.
The full nonnegative zero set is larger. For instance the Horn form has the non-minimal nonnegative zeros
$(a,0,b,0,c)$ with $b=a+c$, on the support $\{1,3,5\}$ where $\horn$ restricts to a rank-one form. Additional
zeros only enlarge the forced kernel, so the rank bound above is unaffected.
\end{remark}

The forced kernel is why the standard rational-rounding step fails on the raw scalings. The numerical solves
return Gram matrices with eigenvalues at the scale of the solver tolerance, and repairing a rounded
near-singular matrix to satisfy the coefficient constraints exactly, in the manner of
\citet{PeyrlParrilo2008}, then produces small negative eigenvalues that no rounding denominator removes: by
Proposition~\ref{prop:kernel} the exact feasible set contains no positive definite point to round toward.

The resolution is not a better rounding scheme but a better target. Shifting the matrix along $\ones$ moves
it into the interior of $\Kc{5}{2}$ while Lemma~\ref{lem:threshold} controls exactly how far one may shift
before the lower cone swallows the candidate. Inside the window, positive definite Gram matrices exist, and
rounding becomes routine. The window is narrow, of width about $10^{-3}$ by \eqref{eq:measured}, which is
why the shift must be chosen against the measured thresholds and cannot be guessed at scale $10^{-1}$.

\subsection{Search}\label{sec:method-search}

Certificates were located by semidefinite programming in the parity-reduced coordinates of
Lemma~\ref{lem:parity}: at level two the $70\times70$ Gram variable splits into one $15\times15$ block, ten
$5\times5$ blocks and five scalars, and at level three the $126\times126$ variable splits into five
$15\times15$ blocks, ten $5\times5$ blocks and one scalar. Two details decide success near the boundary.

First, the objective. A feasibility solve, or a minimum-trace solve, tends to return an extreme point of the
feasible region, which is a low-rank Gram matrix, exactly the object that resists rounding. Maximizing the
smallest eigenvalue instead, in the form $\max\,t$ subject to every block $\succeq tI$, returns a
well-centred Gram whose distance to the semidefinite boundary equals the margin $t>0$. Rounding then has
room to move. On the interior targets $M_0+\eps\ones$ the achieved margins were of order $10^{-1}$ after
normalization, and rounding at denominators of a few thousand succeeded at once.

Second, the exact repair, which the block structure makes trivial. After rounding, the polynomial identity
\eqref{eq:gramA} or \eqref{eq:gramAprime} holds only up to a small rational residual that must be eliminated
exactly. The following observation reduces the correction to a division and applies to any
sum-of-squares certification in a monomial basis, well beyond the present setting.

\begin{proposition}[Elementwise exact repair]\label{prop:repair}
Let $\mathcal{A}$ be the linear map sending a symmetric Gram matrix $Q$, written in a monomial basis, to the
coefficient vector of $z^{\top}Qz$, and suppose each product $z_\alpha z_\beta$ of basis monomials equals a
single target monomial. Then $\mathcal{A}\mathcal{A}^{\top}$ is diagonal, and for any rational $Q_0$ the
minimum-norm exact solution of $\mathcal{A}\,q=b$ nearest $Q_0$ is
\[
  q\;=\;q_0+\mathcal{A}^{\top}\Lambda^{-1}(b-\mathcal{A}q_0),
  \qquad \Lambda=\mathrm{diag}(\mathcal{A}\mathcal{A}^{\top}),
\]
a componentwise operation whose cost is the number of nonzero coefficients.
\end{proposition}

\begin{proof}
The hypothesis says each column of $\mathcal{A}$ (indexed by a Gram entry, i.e.\ an unordered pair of basis
monomials) has exactly one nonzero, in the row of the monomial $z_\alpha z_\beta$. Its value is the
coefficient-map weight of that entry, namely $1$ for a diagonal entry $\alpha=\beta$ and $2$ for an
off-diagonal entry $\alpha\ne\beta$ (since $z_\alpha z_\beta$ then arises from both ordered pairs and appears
with coefficient $2$ in $z^{\top}Qz$). Then
$(\mathcal{A}\mathcal{A}^{\top})_{\mu\nu}=\sum_{\text{columns } c}\mathcal{A}_{\mu c}\mathcal{A}_{\nu c}$
vanishes for $\mu\ne\nu$, since no column has nonzeros in two rows. So $\mathcal{A}\mathcal{A}^{\top}=\Lambda$
is diagonal with $\Lambda_{\mu\mu}$ the \emph{weighted} collision count at $\mu$. This is the sum over the Gram
entries feeding $\mu$ of the squared weight, namely $1$ per diagonal entry and $4$ per off-diagonal entry, an
integer. The stated $q$ is the
image of the normal-equations solution
$q=q_0+\mathcal{A}^{\top}(\mathcal{A}\mathcal{A}^{\top})^{-1}(b-\mathcal{A}q_0)$, which is the exact
minimum-norm correction, and $\Lambda^{-1}$ is a reciprocal of integers.
\end{proof}

The condition holds in our coordinates because a Gram entry pairs two monomials of degree $r+2$ whose product
is one monomial of degree $2r+4$. Proposition~\ref{prop:repair} is what makes exact repair at level three,
with $126$ basis monomials and the coefficient constraints of degree $10$, a matter of seconds. The naive
alternative, an exact solve of a dense rational linear system, is what had made the level-three certificate
appear out of reach.

Dual certificates were found the same way. The functional $y$ is the variable, the constraint is
$\Mom(y)\psd$ blockwise, and the objective maximizes the margin of $\Mom(y)$ under the normalization
$\inner{y}{c_r(\cdot)}=-1$ on the target, after which rounding and exact evaluation of the pairing complete
the certificate.

\subsection{Scale of the certificates}\label{sec:method-scale}

The objects grow quickly with the level, which is why the exact-repair shortcut of
Proposition~\ref{prop:repair} matters. At level $r$ the primal Gram matrix and the dual moment matrix that
excludes $\Kc{5}{r}$ are both indexed by the $\binom{r+6}{r+2}$ monomials of degree $r+2$ in five variables;
Lemma~\ref{lem:parity} splits each into parity blocks whose sizes are the number of degree-$(r+2)$ monomials
in each of the $2^5$ sign classes. A separation theorem $\Kc{5}{r}\subsetneq\Kc{5}{r+1}$ thus pairs a
membership Gram at level $r+1$ with an exclusion moment matrix at level $r$, adjacent rows of
Table~\ref{tab:scale}, which records the sizes.

\begin{table}[htp]\centering
\caption{Basis and block sizes by level. The primal Gram basis and the dual moment basis are the
$\binom{r+6}{4}$ monomials of degree $r+2$. The largest parity block has size $\binom{(r+2-|s|)/2+4}{4}$ with
$|s|$ the least admissible sign-weight ($0$ if $r$ even, $1$ if $r$ odd). The coefficient identity has
$\binom{2r+8}{4}$ monomial equations of degree $2r+4$, of which only the $\binom{r+6}{4}$ even ones are
nontrivial (the parity blocks produce no odd monomials).}\label{tab:scale}
\begin{tabular}{lccccc}
\toprule
level $r$ & primal basis & largest block & coeff.\ constraints & dual basis & used in\\
\midrule
$1$ & $35$  & $5$  & $210$  & $35$  & Thms.~\ref{thm:A},~\ref{thm:Tsdp} (excl.), \S\ref{sec:prelim-worked}\\
$2$ & $70$  & $15$ & $495$  & $70$  & Thms.~\ref{thm:A},~\ref{thm:B},~\ref{thm:Tsdp} (memb.), \ref{thm:Aprime} (excl.)\\
$3$ & $126$ & $15$ & $1001$ & $126$ & Thm.~\ref{thm:Aprime} (memb.), \ref{thm:Adoubleprime} (excl.)\\
$4$ & $210$ & $35$ & $1820$ & $210$ & Thm.~\ref{thm:Adoubleprime} (memb.)\\
\bottomrule
\end{tabular}
\end{table}

Two features of the table explain where the method stops being routine, and how the fourth step is brought
back inside it. The largest parity block grows with the level, through sizes $5,15,15,35$ for $r=1,\dots,4$,
and the rational entries of a well-centred Gram matrix carry denominators that grow with it: the archived
level-three membership Gram already has entries over a common denominator of order $10^{12}$, and the
level-four certificates larger still, each an explicit rational verified exactly. More decisively, along the all-ones
direction the certified interval of shifts narrows with the level: the windows $(\eps_{r+1},\eps_r)$ read
from Table~\ref{tab:landscape} shrink from about $4\cdot10^{-3}$ at the first step to about $2\cdot10^{-4}$ at
the third, and along $\ones$ the fourth-step window closes to the point where the margin for rounding the
size-$210$ objects is comparable to the rounding error. Theorem~\ref{thm:Adoubleprime} overcomes this in two
ways, both instances of the same idea of adapting the geometry to the Horn zeros rather than to the ambient
coordinates. First, the shift is taken along the rank-one direction $B=dd^{\top}$ of
Lemma~\ref{lem:interiordir} rather than $\ones$. Equalizing $\pA{B}$ at the five minimal zeros of $\pA{M_0''}$ widens
the fourth-step window by an order of magnitude and improves the conditioning of both sides. Second, each
certificate is rounded in a metric adapted to its interior reference, so that the margin is measured where it
is actually available. That reference is the coefficient-diagonal Gram $G_B$ for the primal membership Gram,
and a moment functional supported near the Horn zeros for the dual. With these two changes the level-four separation is
certified exactly, primal and dual alike.

\subsection{Verification}\label{sec:method-verify}

A standalone program independent of the search code checks every certificate claim. This covers the membership
Grams and exclusion functionals of Theorems~\ref{thm:A}, \ref{thm:Aprime}, \ref{thm:Adoubleprime},
\ref{thm:B} and~\ref{thm:Tsdp}, the pairings \eqref{eq:pairingA}, \eqref{eq:pairingAprime} and
\eqref{eq:pairingAdoubleprime}, and the state matrix $X$ of \eqref{eq:Xstate} in Corollary~\ref{cor:dicke}. It
reads only the archived certificate files, reconstructs $\horn$ and the scaled/shifted matrices from
Conventions~\ref{conv:horn} and~\ref{conv:scaling}, expands both sides of the polynomial identities as exact
rational coefficient lists, verifies them monomial by monomial, evaluates every pairing exactly, and decides
positive semidefiniteness of every Gram block and of the full moment matrices by the exact rational symmetric
elimination of Lemma~\ref{lem:psdtest}. The program uses integer and rational arithmetic only. No
floating-point number and no semidefinite solve occurs in it. Appendix~\ref{app:verifier} documents the input
format and the checks.

Three consistency controls guard the pipeline itself. The numerical level-one membership test reproduces the
exact criterion of Theorem~\ref{thm:LV} on all $25$ grid points of the family \eqref{eq:LVab} that we
screened. The certification pipeline, run on $\horn$ at level one, reproduces Parrilo's classical membership
$\horn\in\Kc{5}{1}$ with an exact certificate, and the dual search there correctly
fails to produce a witness. The second control also has independent value. It re-derives a known theorem
through the identical toolchain that produces the new ones.

\begin{remark}[A trap in reporting rounded certificates]\label{rem:lattice-trap}
The exact repair changes denominators. A Gram matrix rounded on the grid $\tfrac{1}{N}\mathbb{Z}$ does not,
after repair, have entries in $\tfrac{1}{N}\mathbb{Z}$. The correction divides by the diagonal entries of
$\mathcal{A}\mathcal{A}^{\top}$, the weighted collision counts of Proposition~\ref{prop:repair}. Statements of
the form ``all certificate entries have denominator $N$'' should therefore be avoided. We report a common
denominator or a rounding grid only when it has been recomputed from the final archived data, and we label
rounding grids as such. We archive the reduced exact entries themselves.
\end{remark}

\section{Consequences for the stability-number bounds}\label{sec:stability}

The Parrilo hierarchy was introduced to bound copositive programs, and the motivating example is the stability
number $\alpha(G)$ of a graph. De Klerk and Pasechnik \citep{deKlerkPasechnik2002}, building on the
Motzkin--Straus formulation \citep{Motzkin1965}, write $\alpha(G)$ as a copositive program and relax it level
by level:
\begin{equation}\label{eq:dkp}
  \alpha(G)=\min\{\lambda:\ \lambda(I+A_G)-\ones\in\COP_n\},
  \qquad
  \vartheta^{(r)}(G)=\min\{\lambda:\ \lambda(I+A_G)-\ones\in\Kc{n}{r}\},
\end{equation}
where $A_G$ is the adjacency matrix. The bounds satisfy
$\vartheta^{(0)}\ge\vartheta^{(1)}\ge\cdots\ge\alpha(G)$ and converge, and de Klerk and Pasechnik conjectured
exactness at order $\alpha(G)-1$. These sit atop the classical semidefinite bounds, the Lov\'asz theta number
and Schrijver's strengthening \citep{Lovasz1979,Schrijver1979,Grotschel1981}, and later copositive and
sum-of-squares refinements \citep{Pena2007,Gvozdenovic2007,Dukanovic2010}.

The matrix at the heart of this paper is exactly the stability certificate of the five-cycle. Since
$\alpha(C_5)=2$ and, by Convention~\ref{conv:horn}, $\horn=2(I+A_{C_5})-\ones$, evaluating \eqref{eq:dkp} at
$\lambda=2$ gives the Horn matrix. Its position in the hierarchy, $\horn\in\Kc{5}{1}\setminus\Kc{5}{0}$, is
therefore the statement that the level-one bound is already exact while the level-zero bound is not. The
level-zero bound is classical: since $\Kc{5}{0}=\SPN_5$ is the cone of matrices that are a sum of a positive
semidefinite and an entrywise nonnegative matrix, it carries the entrywise nonnegativity strengthening, and
$\vartheta^{(0)}(G)$ coincides with Schrijver's strengthened theta number $\vartheta'(G)$
\citep{LaurentVargas2109}, not with the plain Lov\'asz number in general. For the five-cycle, however,
$\vartheta'(C_5)=\vartheta(C_5)=\sqrt5$, so $\vartheta^{(0)}(C_5)=\sqrt5$. Thus
\[
  \vartheta^{(0)}(C_5)=\sqrt5\approx 2.236,
  \qquad
  \vartheta^{(1)}(C_5)=2=\alpha(C_5),
\]
and the Horn matrix is precisely what closes the gap from $\sqrt5$ to $2$: at $\lambda=2$ the matrix
$\lambda(I+A_{C_5})-\ones=\horn$ leaves $\SPN_5$ but enters $\Kc{5}{1}$. The five-cycle is the smallest graph
where $\SPN$ tightening fails and one level of $\Kc{}{}$ repairs it, and this single instance is what makes
$\Kc{5}{0}\subsetneq\Kc{5}{1}$ visible. Our separators continue the story
one, two and three levels up. They certify that the cones defining
$\vartheta^{(1)},\vartheta^{(2)},\vartheta^{(3)},\vartheta^{(4)}$ are pairwise distinct, so the geometric
refinement that a higher-order bound can exploit is genuinely present at each of the first three steps, not
merely conjectured.

We are careful about what this does and does not give. A gap $\vartheta^{(r)}(G)>\vartheta^{(r+1)}(G)$ for a
specific graph would require a separator of the constrained form $\lambda(I+A_G)-\ones$, whereas ours are Horn
scalings shifted along $\ones$. Cone strictness is necessary for a bound gap but does not exhibit one. What the
present results settle is the prerequisite geometric question, open until now beyond the classical first step:
that the cones do separate. Locating a graph whose stability bound strictly improves at a prescribed high order
is a natural next target, for which the threshold device of Section~\ref{sec:threshold}, restricted to the
affine slice $\{\lambda(I+A_G)-\ones\}$, is the natural tool.

\section{A quantum application}\label{sec:quantum}

The dual objects constructed here have a second life in quantum information theory. Deciding whether a mixed
state is separable is governed by the positive-partial-transpose criterion \citep{Peres1996,Horodecki1996} and,
more finely, by the symmetric-extension (Doherty, Parrilo and Spedalieri) hierarchy of semidefinite tests
\citep{Doherty2002,Doherty2004,Doherty2005,Terhal2003,Navascues2009}. PPT states that are nonetheless entangled
are bound entangled \citep{Werner1989,Horodecki2009,Guhne2009}. We give the state our certificate produces
explicitly. A companion paper \citep{PaperII} develops the quantum hierarchy further.

Gulati, Nechita and Park \citep{GNP2025} identify the cone dual to $\Kc{n}{r}$ with a family of bosonic
quantum states. To an entrywise nonnegative symmetric $P$ associate the two-particle Dicke-diagonal state
\begin{equation}\label{eq:dicke}
  \rho(P)=\frac{1}{\mathbf{1}^{\top}P\mathbf{1}}\Bigl(\sum_i P_{ii}\,|ii\rangle\!\langle ii|
  +\sum_{i<j}2P_{ij}\,|D_{ij}\rangle\!\langle D_{ij}|\Bigr),
  \qquad |D_{ij}\rangle=\tfrac{1}{\sqrt2}\bigl(|ij\rangle+|ji\rangle\bigr).
\end{equation}
Then $\rho(P)$ is a valid state iff $P$ is entrywise nonnegative. It is PPT iff $P$ is also positive
semidefinite (that is, $P$ is doubly nonnegative, $P\in\mathrm{DNN}_5$). It is separable iff $P$ is completely positive. It admits an
$m$-party fully bosonic extension, positive under every partial transposition and with $\rho(P)$ as every
two-body marginal, iff $P\in(\Kc{5}{m-2})^{*}$.

The dual certificate of Theorem~\ref{thm:A} lands exactly in the gap of this correspondence. Define the
symmetric matrix $X$ by $\inner{X}{N}=\inner{y}{c_1(N)}$ for all $N$, equivalently
$X_{ij}=\sum_k y(x_i^2x_j^2x_k^2)$. From the functional $y$ of Theorem~\ref{thm:A} it is, explicitly,
\begin{equation}\label{eq:Xstate}
  X=\frac{1}{24}
  \begin{pmatrix}
    50452763 & 1384716 & 12296527 & 10741810 & 5759469\\
    1384716 & 1160516 & 426585 & 3442835 & 1975318\\
    12296527 & 426585 & 11148542 & 76968 & 10004023\\
    10741810 & 3442835 & 76968 & 14253636 & 1960470\\
    5759469 & 1975318 & 10004023 & 1960470 & 13988165
  \end{pmatrix}.
\end{equation}
Then $X$ is entrywise nonnegative and positive semidefinite (both checked exactly), so
$X\in\mathrm{DNN}_5$ and $\rho(X)$ is a PPT state. Moreover $\inner{X}{M^{\ast}}=-\tfrac{18919}{28800}<0$ with
$M^{\ast}$ copositive shows, through the duality $\mathrm{CP}_5^{*}=\COP_5$, that $X\notin\mathrm{CP}_5$, so $\rho(X)$ is entangled. A PPT entangled state is bound
entangled. Its Schmidt number is exactly two: each $|D_{ij}\rangle$ has Schmidt rank two and each
$|ii\rangle$ rank one, and the nonnegativity of $X$ makes \eqref{eq:dicke} a convex mixture of vectors of
Schmidt rank at most two, while entanglement rules out a fully separable, Schmidt-number-one decomposition.
Finally $X\in(\Kc{5}{1})^{*}\setminus(\Kc{5}{2})^{*}$, since $X$ pairs nonnegatively with $\Kc{5}{1}$ and
negatively with $M^{\ast}\in\Kc{5}{2}$. By the correspondence $\rho(X)$ admits a three-party fully bosonic
PPT extension but no four-party one.

\begin{corollary}\label{cor:dicke}
The Dicke-diagonal state $\rho(X)$ built from the matrix $X$ of \eqref{eq:Xstate} is a rational PPT bound
entangled state on $\mathbb{C}^5\otimes\mathbb{C}^5$, of Schmidt number two, admitting a three-party but not a
four-party fully bosonic PPT symmetric extension.
\end{corollary}

The same reshaping applied to the dual certificate $y'$ of Theorem~\ref{thm:Aprime} produces, by the same
argument, an explicit state one level up.

The standalone verifier of Appendix~\ref{app:verifier} checks all claims about $X$: entrywise nonnegativity,
positive semidefiniteness, and the pairing $\inner{X}{M^{\ast}}=-\tfrac{18919}{28800}$. The corollary is thus
self-contained modulo the cited public correspondence. On the two-boson sector the complete-graph
extension hierarchy is expected to coincide with the Doherty, Parrilo and Spedalieri hierarchy, which would make
$\rho(X)$ separate consecutive orders of the ordinary PPT symmetric-extension test. A companion paper
\citep{PaperII} pursues this direction.

\section{Conclusions}\label{sec:conclusions}

The Parrilo hierarchy over the copositive cone was known to converge for $n=5$ and to be exact at no finite
level, with a single classical strict step and no explicit witness beyond it. We have exhibited explicit
rational separators for the next three steps,
$\Kc{5}{1}\subsetneq\Kc{5}{2}\subsetneq\Kc{5}{3}\subsetneq\Kc{5}{4}$, each certified by a positive
semidefinite Gram matrix and a dual moment functional and re-verified in integer arithmetic. A one-parameter
family and an inscribed ball show the first gap is full-dimensional. A short combination of two known theorems shows
that strict adjacent steps recur at arbitrarily large levels. The separators were all found by the same
one-dimensional device, the threshold $\eps_r(M)$ along an interior direction, whose strict drops between
levels mark the windows where separators live. The first two new separators use the all-ones direction and the
third a rank-one member $dd^{\top}$ of the same family (Lemma~\ref{lem:interiordir}), chosen because it
equalizes the shift polynomial at the Horn zeros and so keeps both certificates well conditioned. In the
elementary coefficient hierarchy the same device gives closed-form separators at the first seven levels, with
certificates one reads off a coefficient list.

The present results make three questions concrete. First, whether the entry-level map
$\iota(D)=\min\{r:D\horn D\in\Kc{5}{r}\}$ attains every positive integer, which would upgrade
Theorem~\ref{thm:D} from ``infinitely often'' to ``at every step''. The threshold landscape of
Section~\ref{sec:landscape} exhibits the first several values and the tools to test the next. Second, whether
some graph has a stability bound $\vartheta^{(r)}$ that strictly improves at a prescribed high order, for which
the threshold device restricted to the affine slice $\{\lambda(I+A_G)-\ones\}$ is the natural instrument.
Third, whether the explicit chain can be continued to $\Kc{5}{5}$ and beyond by the same means: the
zero-adapted direction and the reference-metric rounding of Section~\ref{sec:thmAdoubleprime} removed the
fourth-level obstruction, but the parity blocks and denominators grow with each level, and whether a single
uniform construction certifies every step remains open. Such a construction would make the recurrence of
Theorem~\ref{thm:D} constructive at every level rather than infinitely often.

\section*{Statements and declarations}\label{sec:decl}

\paragraph{Availability of data and materials.}
All computational claims of the paper are backed by machine-readable archives with a standalone verifier
(Appendix~\ref{app:verifier}). These cover the certificates underlying Theorems~\ref{thm:A}, \ref{thm:Aprime},
\ref{thm:Adoubleprime}, \ref{thm:B} and \ref{thm:Tsdp} (scalings, shifts, exact rational Gram matrices and
exact rational dual moment functionals, and the exact pairings), the state matrix $X$ of Corollary~\ref{cor:dicke},
the monomial-by-monomial coefficient tables behind the closed forms of Theorems~\ref{thm:coef}
and~\ref{thm:Tsep}, and the modelling scripts, solver versions and raw outputs behind the numerical thresholds
of Section~\ref{sec:landscape}. A \texttt{README} lists, for each file, its \texttt{SHA-256} digest, its
schema, and the single command that re-checks it. The verifier depends only on the language's standard library
(exact integer and rational arithmetic) and prints \texttt{PASS} for every theorem. The archive is deposited
in a public repository under the permanent identifier \citep{CertArchive2026}
(\texttt{doi:10.5281/zenodo.22134333}). An identical frozen copy also accompanies the submission as reviewer
material.

\paragraph{Competing interests.}
The authors declare no competing interests.

\begin{appendices}
\section{Certificates for Theorem~\ref{thm:A}}\label{app:certA}
The matrix of Theorem~\ref{thm:A} is $M^{\ast}=M_0+\tfrac{1}{300}\ones$ with $M_0=D\horn D$,
$D=\mathrm{diag}(\tfrac14,4,1,1,1)$. Explicitly,

\[M^{\ast}=\frac{1}{1200}\,\begin{pmatrix}79 & 1204 & -296 & -296 & 304\\ 1204 & 19204 & 4804 & -4796 & -4796\\ -296 & 4804 & 1204 & 1204 & -1196\\ -296 & -4796 & 1204 & 1204 & 1204\\ 304 & -4796 & -1196 & 1204 & 1204\end{pmatrix}\]

The exclusion functional $y$ of \eqref{eq:pairingA} is supported on the monomials of
degree six that are even in every variable. Table~\ref{tab:dualy} lists all nonzero values. Values on all
other degree-six monomials are zero.

\begin{table}[htp]\centering\caption{The functional $y$: values on even degree-6 monomials. All entries share the denominator $q=24$. The table lists the integers $q\,y(m)$.}\label{tab:dualy}
\footnotesize\setlength\tabcolsep{3pt}
\begin{tabular}{lrlrlr}\toprule monomial & $q\,y$ & monomial & $q\,y$ & monomial & $q\,y$\\ \midrule
$x_{5}^{6}$ & $6361094$ & $x_{2}^{2}x_{4}^{4}$ & $1883574$ & $x_{1}^{2}x_{3}^{2}x_{4}^{2}$ & $56499$\\
$x_{4}^{2}x_{5}^{4}$ & $666275$ & $x_{2}^{2}x_{3}^{2}x_{5}^{2}$ & $245003$ & $x_{1}^{2}x_{3}^{4}$ & $3261735$\\
$x_{4}^{4}x_{5}^{2}$ & $901176$ & $x_{2}^{2}x_{3}^{2}x_{4}^{2}$ & $1411$ & $x_{1}^{2}x_{2}^{2}x_{5}^{2}$ & $147649$\\
$x_{4}^{6}$ & $7610948$ & $x_{2}^{2}x_{3}^{4}$ & $146730$ & $x_{1}^{2}x_{2}^{2}x_{4}^{2}$ & $656490$\\
$x_{3}^{2}x_{5}^{4}$ & $3965777$ & $x_{2}^{4}x_{5}^{2}$ & $274472$ & $x_{1}^{2}x_{2}^{2}x_{3}^{2}$ & $1339$\\
$x_{3}^{2}x_{4}^{2}x_{5}^{2}$ & $4011$ & $x_{2}^{4}x_{4}^{2}$ & $519692$ & $x_{1}^{2}x_{2}^{4}$ & $145063$\\
$x_{3}^{2}x_{4}^{4}$ & $8677$ & $x_{2}^{4}x_{3}^{2}$ & $32102$ & $x_{1}^{4}x_{5}^{2}$ & $1466840$\\
$x_{3}^{4}x_{5}^{2}$ & $3720085$ & $x_{2}^{6}$ & $189187$ & $x_{1}^{4}x_{4}^{2}$ & $6172220$\\
$x_{3}^{4}x_{4}^{2}$ & $6370$ & $x_{1}^{2}x_{5}^{4}$ & $2068493$ & $x_{1}^{4}x_{3}^{2}$ & $6907807$\\
$x_{3}^{6}$ & $4013622$ & $x_{1}^{2}x_{4}^{2}x_{5}^{2}$ & $7340$ & $x_{1}^{4}x_{2}^{2}$ & $434175$\\
$x_{2}^{2}x_{5}^{4}$ & $926526$ & $x_{1}^{2}x_{4}^{4}$ & $3849261$ & $x_{1}^{6}$ & $35471721$\\
$x_{2}^{2}x_{4}^{2}x_{5}^{2}$ & $381668$ & $x_{1}^{2}x_{3}^{2}x_{5}^{2}$ & $2069147$ &  & \\\bottomrule\end{tabular}\end{table}

\subsection{The membership Gram matrix}
In the parity-block form of Lemma~\ref{lem:parity} the Gram matrix of \eqref{eq:gramA} is block-diagonal with one of size $15$, ten of size $5$ and five of size $1$ on the $16$ parity classes of the degree-$4$ monomials. Its exact rational entries share the common denominator $33966000$. The full matrix is recorded in the archive (\texttt{k5\_certificate\_bundle.json}) and its positive semidefiniteness, together with the identity \eqref{eq:gramA}, is checked by the standalone verifier of Appendix~\ref{app:verifier}.

\section{Certificates for Theorem~\ref{thm:B}}\label{app:certB}
The family of Theorem~\ref{thm:B} is anchored by a single Gram matrix at $\eps_{\mathrm{lo}}=\tfrac1{10000}$,
of the same form as \eqref{eq:gramA}. 
In the parity-block form of Lemma~\ref{lem:parity} the anchor Gram matrix is block-diagonal with one of size $15$, ten of size $5$ and five of size $1$ on the $16$ parity classes of the degree-$4$ monomials. Its exact rational entries share the common denominator $17832150000$. The full matrix is recorded in the archive (\texttt{k5\_certificate\_bundle.json}) and its positive semidefiniteness, together with the membership identity, is checked by the standalone verifier of Appendix~\ref{app:verifier}.

\section{Certificates for Theorem~\ref{thm:Aprime}}\label{app:certAprime}
The matrix is $M'=M_0'+\tfrac1{800}\ones$ with $M_0'=D'\horn D'$, $D'=\mathrm{diag}(\tfrac1{16},64,1,1,1)$:

\[M'=\frac{1}{6400}\,\begin{pmatrix}33 & 25608 & -392 & -392 & 408\\ 25608 & 26214408 & 409608 & -409592 & -409592\\ -392 & 409608 & 6408 & 6408 & -6392\\ -392 & -409592 & 6408 & 6408 & 6408\\ 408 & -409592 & -6392 & 6408 & 6408\end{pmatrix}\]

The exclusion functional $y'$ of \eqref{eq:pairingAprime} on the even degree-8 monomials
(zero elsewhere):

\begin{table}[htp]\centering\caption{The functional $y'$: values on even degree-8 monomials. All entries share the denominator $q=12$. The table lists the integers $q\,y'(m)$.}\label{tab:dualyp}
\footnotesize\setlength\tabcolsep{3pt}
\begin{tabular}{lrlrlr}\toprule monomial & $q\,y'$ & monomial & $q\,y'$ & monomial & $q\,y'$\\ \midrule
$x_{5}^{8}$ & $99066851271$ & $x_{2}^{2}x_{3}^{6}$ & $7619594$ & $x_{1}^{2}x_{2}^{2}x_{3}^{2}x_{5}^{2}$ & $1190715$\\
$x_{4}^{2}x_{5}^{6}$ & $2774502039$ & $x_{2}^{4}x_{5}^{4}$ & $19357242$ & $x_{1}^{2}x_{2}^{2}x_{3}^{2}x_{4}^{2}$ & $24001$\\
$x_{4}^{4}x_{5}^{4}$ & $1263828614$ & $x_{2}^{4}x_{4}^{2}x_{5}^{2}$ & $1957893$ & $x_{1}^{2}x_{2}^{2}x_{3}^{4}$ & $360580$\\
$x_{4}^{6}x_{5}^{2}$ & $2729316458$ & $x_{2}^{4}x_{4}^{4}$ & $35204997$ & $x_{1}^{2}x_{2}^{4}x_{5}^{2}$ & $566794$\\
$x_{4}^{8}$ & $141601340725$ & $x_{2}^{4}x_{3}^{2}x_{5}^{2}$ & $498032$ & $x_{1}^{2}x_{2}^{4}x_{4}^{2}$ & $6400047$\\
$x_{3}^{2}x_{5}^{6}$ & $26213726887$ & $x_{2}^{4}x_{3}^{2}x_{4}^{2}$ & $2889$ & $x_{1}^{2}x_{2}^{4}x_{3}^{2}$ & $24362$\\
$x_{3}^{2}x_{4}^{2}x_{5}^{4}$ & $3219747$ & $x_{2}^{4}x_{3}^{4}$ & $105523$ & $x_{1}^{2}x_{2}^{6}$ & $115422$\\
$x_{3}^{2}x_{4}^{4}x_{5}^{2}$ & $1242258$ & $x_{2}^{6}x_{5}^{2}$ & $330382$ & $x_{1}^{4}x_{5}^{4}$ & $16838687649$\\
$x_{3}^{2}x_{4}^{6}$ & $3270632$ & $x_{2}^{6}x_{4}^{2}$ & $574447$ & $x_{1}^{4}x_{4}^{2}x_{5}^{2}$ & $7781814$\\
$x_{3}^{4}x_{5}^{4}$ & $23992309345$ & $x_{2}^{6}x_{3}^{2}$ & $8682$ & $x_{1}^{4}x_{4}^{4}$ & $28805121826$\\
$x_{3}^{4}x_{4}^{2}x_{5}^{2}$ & $944618$ & $x_{2}^{8}$ & $25619$ & $x_{1}^{4}x_{3}^{2}x_{5}^{2}$ & $14979236644$\\
$x_{3}^{4}x_{4}^{4}$ & $1040251$ & $x_{1}^{2}x_{5}^{6}$ & $9047577636$ & $x_{1}^{4}x_{3}^{2}x_{4}^{2}$ & $4990894$\\
$x_{3}^{6}x_{5}^{2}$ & $23836057320$ & $x_{1}^{2}x_{4}^{2}x_{5}^{4}$ & $33319446$ & $x_{1}^{4}x_{3}^{4}$ & $18724028714$\\
$x_{3}^{6}x_{4}^{2}$ & $1407433$ & $x_{1}^{2}x_{4}^{4}x_{5}^{2}$ & $82058906$ & $x_{1}^{4}x_{2}^{2}x_{5}^{2}$ & $62082043$\\
$x_{3}^{8}$ & $24447118707$ & $x_{1}^{2}x_{4}^{6}$ & $29017216312$ & $x_{1}^{4}x_{2}^{2}x_{4}^{2}$ & $370881391$\\
$x_{2}^{2}x_{5}^{6}$ & $1207583154$ & $x_{1}^{2}x_{3}^{2}x_{5}^{4}$ & $8142740388$ & $x_{1}^{4}x_{2}^{2}x_{3}^{2}$ & $182281$\\
$x_{2}^{2}x_{4}^{2}x_{5}^{4}$ & $63156225$ & $x_{1}^{2}x_{3}^{2}x_{4}^{2}x_{5}^{2}$ & $799580$ & $x_{1}^{4}x_{2}^{4}$ & $5616228$\\
$x_{2}^{2}x_{4}^{4}x_{5}^{2}$ & $62211041$ & $x_{1}^{2}x_{3}^{2}x_{4}^{4}$ & $1327748$ & $x_{1}^{6}x_{5}^{2}$ & $31880848867$\\
$x_{2}^{2}x_{4}^{6}$ & $2222284556$ & $x_{1}^{2}x_{3}^{4}x_{5}^{2}$ & $9042426812$ & $x_{1}^{6}x_{4}^{2}$ & $93993690866$\\
$x_{2}^{2}x_{3}^{2}x_{5}^{4}$ & $47093412$ & $x_{1}^{2}x_{3}^{4}x_{4}^{2}$ & $835203$ & $x_{1}^{6}x_{3}^{2}$ & $104728178167$\\
$x_{2}^{2}x_{3}^{2}x_{4}^{2}x_{5}^{2}$ & $59688$ & $x_{1}^{2}x_{3}^{6}$ & $10135339940$ & $x_{1}^{6}x_{2}^{2}$ & $536610054$\\
$x_{2}^{2}x_{3}^{2}x_{4}^{4}$ & $67259$ & $x_{1}^{2}x_{2}^{2}x_{5}^{4}$ & $31641444$ & $x_{1}^{8}$ & $1547784497286$\\
$x_{2}^{2}x_{3}^{4}x_{5}^{2}$ & $15584415$ & $x_{1}^{2}x_{2}^{2}x_{4}^{2}x_{5}^{2}$ & $1779776$ &  & \\
$x_{2}^{2}x_{3}^{4}x_{4}^{2}$ & $24550$ & $x_{1}^{2}x_{2}^{2}x_{4}^{4}$ & $429079804$ &  & \\\bottomrule\end{tabular}\end{table}

\subsection{The membership Gram matrix}
In the parity-block form of Lemma~\ref{lem:parity} the Gram matrix of \eqref{eq:gramAprime} is block-diagonal with five of size $15$, ten of size $5$ and one of size $1$ on the $16$ parity classes of the degree-$5$ monomials. Its exact rational entries share the common denominator $1526509152000$. The full matrix is recorded in the archive (\texttt{k5\_level3\_bundle.json}) and its positive semidefiniteness, together with the identity \eqref{eq:gramAprime}, is checked by the standalone verifier of Appendix~\ref{app:verifier}.

\section{Certificate for Theorem~\ref{thm:Tsdp}}\label{app:certT}
The matrix is $M_T=D\,T\,D+\tfrac1{1000}\ones$ with $D=\mathrm{diag}(\tfrac16,6,1,1,1)$ and $T$ the rational
$T(\psi)$ matrix \eqref{eq:Tmat}. Explicitly, $DTD$ is

\[D\,T\,D=\frac{1}{900}\,\begin{pmatrix}25 & -810 & 93 & 93 & -135\\ -810 & 32400 & -4860 & 3348 & 3348\\ 93 & -4860 & 900 & -810 & 558\\ 93 & 3348 & -810 & 900 & -810\\ -135 & 3348 & 558 & -810 & 900\end{pmatrix}\]

The exclusion functional $y_T$ on the even degree-6 monomials (zero elsewhere), with pairing
$\inner{y_T}{c_1(M_T)}=-\tfrac{69223}{108000}$:

\begin{table}[htp]\centering\caption{The functional $y_T$ on even degree-6 monomials. All entries share the denominator $q=12$. The table lists the integers $q\,y_T(m)$.}\label{tab:dualyT}
\footnotesize\setlength\tabcolsep{3pt}
\begin{tabular}{lrlrlr}\toprule monomial & $q\,y_T$ & monomial & $q\,y_T$ & monomial & $q\,y_T$\\ \midrule
$x_{5}^{6}$ & $4242410$ & $x_{2}^{2}x_{4}^{4}$ & $350113$ & $x_{1}^{2}x_{3}^{2}x_{4}^{2}$ & $2452684$\\
$x_{4}^{2}x_{5}^{4}$ & $4834178$ & $x_{2}^{2}x_{3}^{2}x_{5}^{2}$ & $1618$ & $x_{1}^{2}x_{3}^{4}$ & $5587422$\\
$x_{4}^{4}x_{5}^{2}$ & $6444491$ & $x_{2}^{2}x_{3}^{2}x_{4}^{2}$ & $629111$ & $x_{1}^{2}x_{2}^{2}x_{5}^{2}$ & $19818$\\
$x_{4}^{6}$ & $10744262$ & $x_{2}^{2}x_{3}^{4}$ & $1324642$ & $x_{1}^{2}x_{2}^{2}x_{4}^{2}$ & $3195$\\
$x_{3}^{2}x_{5}^{4}$ & $1905038$ & $x_{2}^{4}x_{5}^{2}$ & $534$ & $x_{1}^{2}x_{2}^{2}x_{3}^{2}$ & $800081$\\
$x_{3}^{2}x_{4}^{2}x_{5}^{2}$ & $3025580$ & $x_{2}^{4}x_{4}^{2}$ & $64819$ & $x_{1}^{2}x_{2}^{4}$ & $203971$\\
$x_{3}^{2}x_{4}^{4}$ & $7599833$ & $x_{2}^{4}x_{3}^{2}$ & $163762$ & $x_{1}^{4}x_{5}^{2}$ & $23438137$\\
$x_{3}^{4}x_{5}^{2}$ & $1649575$ & $x_{2}^{6}$ & $24167$ & $x_{1}^{4}x_{4}^{2}$ & $4543098$\\
$x_{3}^{4}x_{4}^{2}$ & $8148568$ & $x_{1}^{2}x_{5}^{4}$ & $6652075$ & $x_{1}^{4}x_{3}^{2}$ & $5484896$\\
$x_{3}^{6}$ & $12723968$ & $x_{1}^{2}x_{4}^{2}x_{5}^{2}$ & $4270490$ & $x_{1}^{4}x_{2}^{2}$ & $4037520$\\
$x_{2}^{2}x_{5}^{4}$ & $2094$ & $x_{1}^{2}x_{4}^{4}$ & $5286456$ & $x_{1}^{6}$ & $202459478$\\
$x_{2}^{2}x_{4}^{2}x_{5}^{2}$ & $1643$ & $x_{1}^{2}x_{3}^{2}x_{5}^{2}$ & $12912$ &  & \\\bottomrule\end{tabular}\end{table}

\subsection{The membership Gram matrix for $M_T$}
In the parity-block form of Lemma~\ref{lem:parity} the membership Gram matrix for $M_T$ is block-diagonal with one of size $15$, ten of size $5$ and five of size $1$ on the $16$ parity classes of the degree-$4$ monomials. Its exact rational entries share the common denominator $254745000$. The full matrix is recorded in the archive (\texttt{k5\_Torbit\_bundle.json}) and its positive semidefiniteness, together with the membership identity, is checked by the standalone verifier of Appendix~\ref{app:verifier}.

\section{Certificates for Theorem~\ref{thm:Adoubleprime}}\label{app:certAdoubleprime}
The matrix is $M''=M_0''+\tfrac1{1024}B$ with $M_0''=D''\horn D''$, $D''=\mathrm{diag}(\tfrac1{12},48,1,1,1)$,
and $B=dd^{\top}$ the rank-one interior direction of
Lemma~\ref{lem:interiordir}, $d=(\tfrac1{12},48,1,1,1)$. Explicitly, $M''$ is

\[M''=\frac{1}{147456}\,\begin{pmatrix}1025 & 590400 & -12276 & -12276 & 12300\\ 590400 & 340070400 & 7084800 & -7070976 & -7070976\\ -12276 & 7084800 & 147600 & 147600 & -147312\\ -12276 & -7070976 & 147600 & 147600 & 147600\\ 12300 & -7070976 & -147312 & 147600 & 147600\end{pmatrix}\]

Unlike the earlier certificates, the two Gram/moment objects here are too large to typeset in
full: membership uses a block Gram $Q''$ on the $210$ monomials of degree $6$ (blocks of sizes $35$, ten of
$15$, and five of $5$), and exclusion uses a dual functional $y''$ on the $126$ even monomials of degree
$10$ whose full $126\times126$ moment matrix is positive semidefinite. Both are exact rationals. The
membership Gram is assembled on a rounding grid of denominator $27720$ together with the reserved
$G_B$-shift below, and the dual on a rounding grid of denominator $504504000$, so their reduced entries
carry substantially larger denominators. Both are stored in the accompanying archive
(\texttt{k5\_level4\_membership.json} and \texttt{k5\_level4\_dual.json}) and checked by the standalone
verifier of Appendix~\ref{app:verifier}, which reconstructs $M''$ from $D''$, $B$ and the Horn matrix,
confirms the coefficient identity \eqref{eq:gramAdoubleprime} and positive semidefiniteness of every Gram
block, and verifies that the full moment matrix is positive semidefinite with the exact pairing
\[
  \inner{y''}{c_3(M'')}\;=\;-\frac{944597628682264591108409}{714164561510400000000000000}\;\approx\;-0.001323\;<\;0,
\]
which by Lemma~\ref{lem:exclusion} places $M''$ outside $\Kc{5}{3}$.

The membership Gram is assembled by the reserved-shift construction of Section~\ref{sec:thmAdoubleprime}: an
exact Gram for $(\sum_i x_i^2)^4\pA{M_0''+\frac1{2048}B}$ is formed by rounding a numerical Gram taken in the
metric of $G_B=\mathrm{diag}(b_\alpha)$ and projecting it exactly onto the coefficient constraints, after
which $\tfrac1{2048}G_B$ is added. Since $G_B$ represents $(\sum_i x_i^2)^4\pA{B}$ exactly, the result
represents $(\sum_i x_i^2)^4\pA{M''}$ exactly and inherits positive semidefiniteness from the reserved
positive-definite summand.

\section{The standalone verifier}\label{app:verifier}
The verifier is a single file using only integer and rational arithmetic from the language's standard
library. Its input is one archive per theorem containing: the scaling $D$, the shift $\eps$, the monomial
lists, the Gram blocks, and the dual functional. It performs, in order: (1) reconstruction of $\horn$ from
Convention~\ref{conv:horn} and of $M_0$, $M^{\ast}$ from Convention~\ref{conv:scaling}, cross-checked
against the archived copies; (2) expansion of $(\sum_ix_i^2)^{r}\pA{M}$ as an exact coefficient list by
convolution of exponent dictionaries; (3) accumulation of $z^{\top}Qz$ from the blocks and comparison with
(2), monomial by monomial; (4) exact evaluation of the pairing $\inner{y}{c_{r-1}(M)}$; (5) positive
semidefiniteness of every Gram block and of the full moment matrix $\Mom(y)$ by the exact test of
Lemma~\ref{lem:psdtest}. The checks reject perturbed data. Altering any Gram or moment entry breaks the
coefficient identity (3) or the semidefiniteness test (5). Runtime is seconds. The verifier and archives accompany the paper.

\begin{lemma}[Exact semidefiniteness test]\label{lem:psdtest}
Let $A\in\Sym^m(\Q)$. Run symmetric Gaussian elimination with the following pivoting. At step $k$ (for
$k=1,\dots,m$), among the rows $i\ge k$ seek one with $A_{ii}>0$. If one exists, apply the symmetric
permutation swapping indices $k$ and $i$ (interchanging both the two rows and the two columns, which preserves
symmetry), then apply the symmetric Schur-complement update to the trailing block,
\[
  A_{ij}\;\leftarrow\;A_{ij}-\frac{A_{ik}A_{jk}}{A_{kk}}\qquad(i,j>k),
\]
which is the congruence that simultaneously clears row $k$ and column $k$ below the pivot. If no row $i\ge k$
has $A_{ii}>0$, halt and return \textsc{true} iff the entire trailing
principal submatrix on indices $k,\dots,m$ is zero, and \textsc{false} otherwise (in particular if some
$A_{ii}<0$). If all $m$ steps complete, return \textsc{true}. The procedure is exact over $\Q$ and returns
\textsc{true} if and only if $A$ is positive semidefinite.
\end{lemma}

\begin{proof}
Each elimination step is a congruence $A\mapsto E^{\top}AE$ by a rational elementary matrix $E$, and each
pivot search is a symmetric permutation congruence. Congruences preserve inertia, so they preserve positive
semidefiniteness. If a positive pivot $A_{kk}>0$ is found, the Schur complement of that entry is formed
exactly and the leading direction contributes a strictly positive eigenvalue, consistent with
semidefiniteness. Suppose at some step no row $i\ge k$ has positive diagonal. If the trailing submatrix
$B$ on $k,\dots,m$ is the zero matrix, the elimination is complete and all realized pivots were positive, so
$A$ is congruent to a diagonal matrix with nonnegative entries, hence positive semidefinite. Otherwise $B$
has all diagonal entries $\le0$ and is not identically zero: if some $B_{ii}<0$ then the corresponding unit
vector is a direction of negative curvature; and if all $B_{ii}=0$ but some $B_{ij}\ne0$, the vector
$e_i-\operatorname{sgn}(B_{ij})e_j$ gives $\,^{\!}v^{\top}Bv=-2|B_{ij}|<0$. In both cases $B$, and therefore
$A$, is not positive semidefinite. All arithmetic is over $\Q$, so the decision is exact.
\end{proof}
\end{appendices}

\bibliography{refs}


\begin{thebibliography}{63}
\ifx \bisbn   \undefined \def \bisbn  #1{ISBN #1}\fi
\ifx \binits  \undefined \def \binits#1{#1}\fi
\ifx \bauthor  \undefined \def \bauthor#1{#1}\fi
\ifx \batitle  \undefined \def \batitle#1{#1}\fi
\ifx \bjtitle  \undefined \def \bjtitle#1{#1}\fi
\ifx \bvolume  \undefined \def \bvolume#1{\textbf{#1}}\fi
\ifx \byear  \undefined \def \byear#1{#1}\fi
\ifx \bissue  \undefined \def \bissue#1{#1}\fi
\ifx \bfpage  \undefined \def \bfpage#1{#1}\fi
\ifx \blpage  \undefined \def \blpage #1{#1}\fi
\ifx \burl  \undefined \def \burl#1{\textsf{#1}}\fi
\ifx \doiurl  \undefined \def \doiurl#1{\url{https://doi.org/#1}}\fi
\ifx \betal  \undefined \def \betal{\textit{et al.}}\fi
\ifx \binstitute  \undefined \def \binstitute#1{#1}\fi
\ifx \binstitutionaled  \undefined \def \binstitutionaled#1{#1}\fi
\ifx \bctitle  \undefined \def \bctitle#1{#1}\fi
\ifx \beditor  \undefined \def \beditor#1{#1}\fi
\ifx \bpublisher  \undefined \def \bpublisher#1{#1}\fi
\ifx \bbtitle  \undefined \def \bbtitle#1{#1}\fi
\ifx \bedition  \undefined \def \bedition#1{#1}\fi
\ifx \bseriesno  \undefined \def \bseriesno#1{#1}\fi
\ifx \blocation  \undefined \def \blocation#1{#1}\fi
\ifx \bsertitle  \undefined \def \bsertitle#1{#1}\fi
\ifx \bsnm \undefined \def \bsnm#1{#1}\fi
\ifx \bsuffix \undefined \def \bsuffix#1{#1}\fi
\ifx \bparticle \undefined \def \bparticle#1{#1}\fi
\ifx \barticle \undefined \def \barticle#1{#1}\fi
\bibcommenthead
\ifx \bconfdate \undefined \def \bconfdate #1{#1}\fi
\ifx \botherref \undefined \def \botherref #1{#1}\fi
\ifx \url \undefined \def \url#1{\textsf{#1}}\fi
\ifx \bchapter \undefined \def \bchapter#1{#1}\fi
\ifx \bbook \undefined \def \bbook#1{#1}\fi
\ifx \bcomment \undefined \def \bcomment#1{#1}\fi
\ifx \oauthor \undefined \def \oauthor#1{#1}\fi
\ifx \citeauthoryear \undefined \def \citeauthoryear#1{#1}\fi
\ifx \endbibitem  \undefined \def \endbibitem {}\fi
\ifx \bconflocation  \undefined \def \bconflocation#1{#1}\fi
\ifx \arxivurl  \undefined \def \arxivurl#1{\textsf{#1}}\fi
\csname PreBibitemsHook\endcsname

\bibitem[\protect\citeauthoryear{{de Klerk} and
  Pasechnik}{2002}]{deKlerkPasechnik2002}
\begin{barticle}
\bauthor{\bsnm{{de Klerk}}, \binits{E.}},
\bauthor{\bsnm{Pasechnik}, \binits{D.V.}}:
\batitle{Approximation of the stability number of a graph via copositive
  programming}.
\bjtitle{SIAM Journal on Optimization}
\bvolume{12}(\bissue{4}),
\bfpage{875}--\blpage{892}
(\byear{2002})
\doiurl{10.1137/S1052623401383248}
\end{barticle}
\endbibitem

\bibitem[\protect\citeauthoryear{Bomze and {de Klerk}}{2002}]{BomzeDeKlerk2002}
\begin{barticle}
\bauthor{\bsnm{Bomze}, \binits{I.M.}},
\bauthor{\bsnm{{de Klerk}}, \binits{E.}}:
\batitle{Solving standard quadratic optimization problems via linear,
  semidefinite and copositive programming}.
\bjtitle{Journal of Global Optimization}
\bvolume{24}(\bissue{2}),
\bfpage{163}--\blpage{185}
(\byear{2002})
\doiurl{10.1023/A:1020209017701}
\end{barticle}
\endbibitem

\bibitem[\protect\citeauthoryear{Burer}{2009}]{Burer2009}
\begin{barticle}
\bauthor{\bsnm{Burer}, \binits{S.}}:
\batitle{On the copositive representation of binary and continuous nonconvex
  quadratic programs}.
\bjtitle{Mathematical Programming}
\bvolume{120}(\bissue{2}),
\bfpage{479}--\blpage{495}
(\byear{2009})
\doiurl{10.1007/s10107-008-0223-z}
\end{barticle}
\endbibitem

\bibitem[\protect\citeauthoryear{D{\"u}r}{2010}]{Dur2010}
\begin{bchapter}
\bauthor{\bsnm{D{\"u}r}, \binits{M.}}:
\bctitle{Copositive programming: a survey}.
In: \beditor{\bsnm{Diehl}, \binits{M.}},
\beditor{\bsnm{Glineur}, \binits{F.}},
\beditor{\bsnm{Jarlebring}, \binits{E.}},
\beditor{\bsnm{Michiels}, \binits{W.}} (eds.)
\bbtitle{Recent Advances in Optimization and Its Applications in Engineering},
pp. \bfpage{3}--\blpage{20}.
\bpublisher{Springer},
\blocation{Berlin, Heidelberg}
(\byear{2010}).
\doiurl{10.1007/978-3-642-12598-0_1}
\end{bchapter}
\endbibitem

\bibitem[\protect\citeauthoryear{Murty and Kabadi}{1987}]{MurtyKabadi1987}
\begin{barticle}
\bauthor{\bsnm{Murty}, \binits{K.G.}},
\bauthor{\bsnm{Kabadi}, \binits{S.N.}}:
\batitle{Some {NP}-complete problems in quadratic and nonlinear programming}.
\bjtitle{Mathematical Programming}
\bvolume{39}(\bissue{2}),
\bfpage{117}--\blpage{129}
(\byear{1987})
\doiurl{10.1007/BF02592948}
\end{barticle}
\endbibitem

\bibitem[\protect\citeauthoryear{Parrilo}{2000}]{Parrilo2000}
\begin{botherref}
\oauthor{\bsnm{Parrilo}, \binits{P.A.}}:
Structured semidefinite programs and semialgebraic geometry methods in
  robustness and optimization.
PhD thesis,
California Institute of Technology
(2000).
\doiurl{10.7907/2K6Y-CH43}
\end{botherref}
\endbibitem

\bibitem[\protect\citeauthoryear{Diananda}{1962}]{Diananda1962}
\begin{barticle}
\bauthor{\bsnm{Diananda}, \binits{P.H.}}:
\batitle{On non-negative forms in real variables some or all of which are
  non-negative}.
\bjtitle{Proceedings of the Cambridge Philosophical Society}
\bvolume{58}(\bissue{1}),
\bfpage{17}--\blpage{25}
(\byear{1962})
\doiurl{10.1017/S0305004100036185}
\end{barticle}
\endbibitem

\bibitem[\protect\citeauthoryear{Dickinson et~al.}{2013}]{DDGH2013}
\begin{barticle}
\bauthor{\bsnm{Dickinson}, \binits{P.J.C.}},
\bauthor{\bsnm{D{\"u}r}, \binits{M.}},
\bauthor{\bsnm{Gijben}, \binits{L.}},
\bauthor{\bsnm{Hildebrand}, \binits{R.}}:
\batitle{Scaling relationship between the copositive cone and {P}arrilo's first
  level approximation}.
\bjtitle{Optimization Letters}
\bvolume{7}(\bissue{8}),
\bfpage{1669}--\blpage{1679}
(\byear{2013})
\doiurl{10.1007/s11590-012-0523-3}
\end{barticle}
\endbibitem

\bibitem[\protect\citeauthoryear{Schweighofer and
  Vargas}{2024}]{SchweighoferVargas2024}
\begin{barticle}
\bauthor{\bsnm{Schweighofer}, \binits{M.}},
\bauthor{\bsnm{Vargas}, \binits{L.F.}}:
\batitle{Sum-of-squares certificates for copositivity via test states}.
\bjtitle{SIAM Journal on Applied Algebra and Geometry}
\bvolume{8}(\bissue{4}),
\bfpage{797}--\blpage{820}
(\byear{2024})
\doiurl{10.1137/23M1611798}
\end{barticle}
\endbibitem

\bibitem[\protect\citeauthoryear{Gulati et~al.}{2025}]{GNP2025}
\begin{botherref}
\oauthor{\bsnm{Gulati}, \binits{A.}},
\oauthor{\bsnm{Nechita}, \binits{I.}},
\oauthor{\bsnm{Park}, \binits{S.-J.}}:
Positive maps and extendibility hierarchies from copositive matrices
(2025).
\doiurl{10.48550/arXiv.2509.15201} .
\url{https://arxiv.org/abs/2509.15201}
\end{botherref}
\endbibitem

\bibitem[\protect\citeauthoryear{Britz and Laurent}{2025}]{BritzLaurent2025}
\begin{botherref}
\oauthor{\bsnm{Britz}, \binits{J.}},
\oauthor{\bsnm{Laurent}, \binits{M.}}:
Semidefinite hierarchies for diagonal unitary invariant bipartite quantum
  states
(2025).
\doiurl{10.48550/arXiv.2512.06551} .
\url{https://arxiv.org/abs/2512.06551}
\end{botherref}
\endbibitem

\bibitem[\protect\citeauthoryear{Bomze et~al.}{2000}]{Bomze2000}
\begin{barticle}
\bauthor{\bsnm{Bomze}, \binits{I.M.}},
\bauthor{\bsnm{D{\"u}r}, \binits{M.}},
\bauthor{\bsnm{{de Klerk}}, \binits{E.}},
\bauthor{\bsnm{Roos}, \binits{C.}},
\bauthor{\bsnm{Quist}, \binits{A.J.}},
\bauthor{\bsnm{Terlaky}, \binits{T.}}:
\batitle{On copositive programming and standard quadratic optimization
  problems}.
\bjtitle{Journal of Global Optimization}
\bvolume{18}(\bissue{4}),
\bfpage{301}--\blpage{320}
(\byear{2000})
\doiurl{10.1023/A:1026583532263}
\end{barticle}
\endbibitem

\bibitem[\protect\citeauthoryear{Motzkin and Straus}{1965}]{Motzkin1965}
\begin{barticle}
\bauthor{\bsnm{Motzkin}, \binits{T.S.}},
\bauthor{\bsnm{Straus}, \binits{E.G.}}:
\batitle{Maxima for graphs and a new proof of a theorem of {T}ur{\'a}n}.
\bjtitle{Canadian Journal of Mathematics}
\bvolume{17},
\bfpage{533}--\blpage{540}
(\byear{1965})
\doiurl{10.4153/CJM-1965-053-6}
\end{barticle}
\endbibitem

\bibitem[\protect\citeauthoryear{Bomze}{2012}]{Bomze2012}
\begin{barticle}
\bauthor{\bsnm{Bomze}, \binits{I.M.}}:
\batitle{Copositive optimization -- recent developments and applications}.
\bjtitle{European Journal of Operational Research}
\bvolume{216}(\bissue{3}),
\bfpage{509}--\blpage{520}
(\byear{2012})
\doiurl{10.1016/j.ejor.2011.04.026}
\end{barticle}
\endbibitem

\bibitem[\protect\citeauthoryear{Burer}{2012}]{Burer2012}
\begin{bchapter}
\bauthor{\bsnm{Burer}, \binits{S.}}:
\bctitle{Copositive programming}.
In: \beditor{\bsnm{Anjos}, \binits{M.F.}},
\beditor{\bsnm{Lasserre}, \binits{J.B.}} (eds.)
\bbtitle{Handbook on Semidefinite, Conic and Polynomial Optimization}.
\bsertitle{International Series in Operations Research \& Management Science},
vol. \bseriesno{166},
pp. \bfpage{201}--\blpage{218}.
\bpublisher{Springer},
\blocation{New York, NY}
(\byear{2012}).
\doiurl{10.1007/978-1-4614-0769-0_8}
\end{bchapter}
\endbibitem

\bibitem[\protect\citeauthoryear{Bomze et~al.}{2012}]{Bomze2012a}
\begin{barticle}
\bauthor{\bsnm{Bomze}, \binits{I.M.}},
\bauthor{\bsnm{Schachinger}, \binits{W.}},
\bauthor{\bsnm{Uchida}, \binits{G.}}:
\batitle{Think co(mpletely)positive! {M}atrix properties, examples and a
  clustered bibliography on copositive optimization}.
\bjtitle{Journal of Global Optimization}
\bvolume{52}(\bissue{3}),
\bfpage{423}--\blpage{445}
(\byear{2012})
\doiurl{10.1007/s10898-011-9749-3}
\end{barticle}
\endbibitem

\bibitem[\protect\citeauthoryear{Hall and Newman}{1963}]{Hall1963}
\begin{barticle}
\bauthor{\bsnm{Hall}, \binits{M.} \bsuffix{Jr.}},
\bauthor{\bsnm{Newman}, \binits{M.}}:
\batitle{Copositive and completely positive quadratic forms}.
\bjtitle{Mathematical Proceedings of the Cambridge Philosophical Society}
\bvolume{59}(\bissue{2}),
\bfpage{329}--\blpage{339}
(\byear{1963})
\doiurl{10.1017/S0305004100036951}
\end{barticle}
\endbibitem

\bibitem[\protect\citeauthoryear{Berman and Shaked-Monderer}{2003}]{Berman2003}
\begin{bbook}
\bauthor{\bsnm{Berman}, \binits{A.}},
\bauthor{\bsnm{Shaked-Monderer}, \binits{N.}}:
\bbtitle{Completely Positive Matrices}.
\bpublisher{World Scientific},
\blocation{River Edge, NJ}
(\byear{2003}).
\doiurl{10.1142/5273}
\end{bbook}
\endbibitem

\bibitem[\protect\citeauthoryear{Dickinson and Gijben}{2014}]{Dickinson2014}
\begin{barticle}
\bauthor{\bsnm{Dickinson}, \binits{P.J.C.}},
\bauthor{\bsnm{Gijben}, \binits{L.}}:
\batitle{On the computational complexity of membership problems for the
  completely positive cone and its dual}.
\bjtitle{Computational Optimization and Applications}
\bvolume{57}(\bissue{2}),
\bfpage{403}--\blpage{415}
(\byear{2014})
\doiurl{10.1007/s10589-013-9594-z}
\end{barticle}
\endbibitem

\bibitem[\protect\citeauthoryear{V{\"a}liaho}{1986}]{Valiaho1986}
\begin{barticle}
\bauthor{\bsnm{V{\"a}liaho}, \binits{H.}}:
\batitle{Criteria for copositive matrices}.
\bjtitle{Linear Algebra and its Applications}
\bvolume{81},
\bfpage{19}--\blpage{34}
(\byear{1986})
\doiurl{10.1016/0024-3795(86)90246-6}
\end{barticle}
\endbibitem

\bibitem[\protect\citeauthoryear{Baumert}{1966}]{Baumert1966}
\begin{barticle}
\bauthor{\bsnm{Baumert}, \binits{L.D.}}:
\batitle{Extreme copositive quadratic forms}.
\bjtitle{Pacific Journal of Mathematics}
\bvolume{19}(\bissue{2}),
\bfpage{197}--\blpage{204}
(\byear{1966})
\doiurl{10.2140/pjm.1966.19.197}
\end{barticle}
\endbibitem

\bibitem[\protect\citeauthoryear{Hildebrand}{2012}]{Hildebrand2012}
\begin{barticle}
\bauthor{\bsnm{Hildebrand}, \binits{R.}}:
\batitle{The extreme rays of the $5\times5$ copositive cone}.
\bjtitle{Linear Algebra and its Applications}
\bvolume{437}(\bissue{7}),
\bfpage{1538}--\blpage{1547}
(\byear{2012})
\doiurl{10.1016/j.laa.2012.04.017}
\end{barticle}
\endbibitem

\bibitem[\protect\citeauthoryear{Hildebrand}{2014}]{Hildebrand2014}
\begin{barticle}
\bauthor{\bsnm{Hildebrand}, \binits{R.}}:
\batitle{Minimal zeros of copositive matrices}.
\bjtitle{Linear Algebra and its Applications}
\bvolume{459},
\bfpage{154}--\blpage{174}
(\byear{2014})
\doiurl{10.1016/j.laa.2014.07.004}
\end{barticle}
\endbibitem

\bibitem[\protect\citeauthoryear{Haynsworth and
  Hoffman}{1969}]{HaynsworthHoffman1969}
\begin{barticle}
\bauthor{\bsnm{Haynsworth}, \binits{E.}},
\bauthor{\bsnm{Hoffman}, \binits{A.J.}}:
\batitle{Two remarks on copositive matrices}.
\bjtitle{Linear Algebra and its Applications}
\bvolume{2}(\bissue{4}),
\bfpage{387}--\blpage{392}
(\byear{1969})
\doiurl{10.1016/0024-3795(69)90011-1}
\end{barticle}
\endbibitem

\bibitem[\protect\citeauthoryear{Laurent and Vargas}{2023}]{LaurentVargas2109}
\begin{barticle}
\bauthor{\bsnm{Laurent}, \binits{M.}},
\bauthor{\bsnm{Vargas}, \binits{L.F.}}:
\batitle{Exactness of {P}arrilo's conic approximations for copositive matrices
  and associated low order bounds for the stability number of a graph}.
\bjtitle{Mathematics of Operations Research}
\bvolume{48}(\bissue{2}),
\bfpage{1017}--\blpage{1043}
(\byear{2023})
\doiurl{10.1287/moor.2022.1290}
{\href{https://arxiv.org/abs/2109.12876}{{arXiv:2109.12876}}}.
\bcomment{arXiv:2109.12876}
\end{barticle}
\endbibitem

\bibitem[\protect\citeauthoryear{Parrilo}{2003}]{Parrilo2003}
\begin{barticle}
\bauthor{\bsnm{Parrilo}, \binits{P.A.}}:
\batitle{Semidefinite programming relaxations for semialgebraic problems}.
\bjtitle{Mathematical Programming}
\bvolume{96}(\bissue{2}),
\bfpage{293}--\blpage{320}
(\byear{2003})
\doiurl{10.1007/s10107-003-0387-5}
\end{barticle}
\endbibitem

\bibitem[\protect\citeauthoryear{Lasserre}{2001}]{Lasserre2001}
\begin{barticle}
\bauthor{\bsnm{Lasserre}, \binits{J.B.}}:
\batitle{Global optimization with polynomials and the problem of moments}.
\bjtitle{SIAM Journal on Optimization}
\bvolume{11}(\bissue{3}),
\bfpage{796}--\blpage{817}
(\byear{2001})
\doiurl{10.1137/S1052623400366802}
\end{barticle}
\endbibitem

\bibitem[\protect\citeauthoryear{Nie}{2014}]{Nie2014}
\begin{barticle}
\bauthor{\bsnm{Nie}, \binits{J.}}:
\batitle{Optimality conditions and finite convergence of {L}asserre's
  hierarchy}.
\bjtitle{Mathematical Programming}
\bvolume{146}(\bissue{1--2}),
\bfpage{97}--\blpage{121}
(\byear{2014})
\doiurl{10.1007/s10107-013-0680-x}
\end{barticle}
\endbibitem

\bibitem[\protect\citeauthoryear{Laurent}{2009}]{Laurent2009}
\begin{bchapter}
\bauthor{\bsnm{Laurent}, \binits{M.}}:
\bctitle{Sums of squares, moment matrices and optimization over polynomials}.
In: \beditor{\bsnm{Putinar}, \binits{M.}},
\beditor{\bsnm{Sullivant}, \binits{S.}} (eds.)
\bbtitle{Emerging Applications of Algebraic Geometry}.
\bsertitle{The IMA Volumes in Mathematics and its Applications},
vol. \bseriesno{149},
pp. \bfpage{157}--\blpage{270}.
\bpublisher{Springer},
\blocation{New York}
(\byear{2009}).
\doiurl{10.1007/978-0-387-09686-5_7}
\end{bchapter}
\endbibitem

\bibitem[\protect\citeauthoryear{Blekherman
  et~al.}{2013}]{BlekhermanParriloThomas2013}
\begin{bbook}
\beditor{\bsnm{Blekherman}, \binits{G.}},
\beditor{\bsnm{Parrilo}, \binits{P.A.}},
\beditor{\bsnm{Thomas}, \binits{R.R.}} (eds.):
\bbtitle{Semidefinite Optimization and Convex Algebraic Geometry}.
\bsertitle{MOS-SIAM Series on Optimization},
vol. \bseriesno{13}.
\bpublisher{SIAM},
\blocation{Philadelphia, PA}
(\byear{2013}).
\doiurl{10.1137/1.9781611972290}
\end{bbook}
\endbibitem

\bibitem[\protect\citeauthoryear{Lasserre}{2010}]{Lasserre2010}
\begin{bbook}
\bauthor{\bsnm{Lasserre}, \binits{J.B.}}:
\bbtitle{Moments, Positive Polynomials and Their Applications}.
\bpublisher{Imperial College Press},
\blocation{London}
(\byear{2010}).
\doiurl{10.1142/p665}
\end{bbook}
\endbibitem

\bibitem[\protect\citeauthoryear{{de Klerk} and
  Laurent}{2019}]{deKlerkLaurent2019}
\begin{bchapter}
\bauthor{\bsnm{{de Klerk}}, \binits{E.}},
\bauthor{\bsnm{Laurent}, \binits{M.}}:
\bctitle{A survey of semidefinite programming approaches to the generalized
  problem of moments and their error analysis}.
In: \beditor{\bsnm{Araujo}, \binits{C.}},
\beditor{\bsnm{Benkart}, \binits{G.}},
\beditor{\bsnm{Praeger}, \binits{C.E.}},
\beditor{\bsnm{Tanbay}, \binits{B.}} (eds.)
\bbtitle{World Women in Mathematics 2018}.
\bsertitle{Association for Women in Mathematics Series},
vol. \bseriesno{20},
pp. \bfpage{17}--\blpage{56}.
\bpublisher{Springer},
\blocation{Cham}
(\byear{2019}).
\doiurl{10.1007/978-3-030-21170-7_1}
\end{bchapter}
\endbibitem

\bibitem[\protect\citeauthoryear{Schm{\"u}dgen}{1991}]{Schmudgen1991}
\begin{barticle}
\bauthor{\bsnm{Schm{\"u}dgen}, \binits{K.}}:
\batitle{The {K}-moment problem for compact semi-algebraic sets}.
\bjtitle{Mathematische Annalen}
\bvolume{289}(\bissue{2}),
\bfpage{203}--\blpage{206}
(\byear{1991})
\doiurl{10.1007/BF01446568}
\end{barticle}
\endbibitem

\bibitem[\protect\citeauthoryear{Putinar}{1993}]{Putinar1993}
\begin{barticle}
\bauthor{\bsnm{Putinar}, \binits{M.}}:
\batitle{Positive polynomials on compact semi-algebraic sets}.
\bjtitle{Indiana University Mathematics Journal}
\bvolume{42}(\bissue{3}),
\bfpage{969}--\blpage{984}
(\byear{1993})
\doiurl{10.1512/iumj.1993.42.42045}
\end{barticle}
\endbibitem

\bibitem[\protect\citeauthoryear{Powers and Reznick}{2001}]{PowersReznick2001}
\begin{barticle}
\bauthor{\bsnm{Powers}, \binits{V.}},
\bauthor{\bsnm{Reznick}, \binits{B.}}:
\batitle{A new bound for {P}\'olya's theorem with applications to polynomials
  positive on polyhedra}.
\bjtitle{Journal of Pure and Applied Algebra}
\bvolume{164}(\bissue{1--2}),
\bfpage{221}--\blpage{229}
(\byear{2001})
\doiurl{10.1016/S0022-4049(00)00155-9}
\end{barticle}
\endbibitem

\bibitem[\protect\citeauthoryear{Reznick}{2000}]{Reznick2000}
\begin{bchapter}
\bauthor{\bsnm{Reznick}, \binits{B.}}:
\bctitle{Some concrete aspects of {H}ilbert's 17th problem}.
In: \beditor{\bsnm{Delzell}, \binits{C.N.}},
\beditor{\bsnm{Madden}, \binits{J.J.}} (eds.)
\bbtitle{Real Algebraic Geometry and Ordered Structures}.
\bsertitle{Contemporary Mathematics},
vol. \bseriesno{253},
pp. \bfpage{251}--\blpage{272}.
\bpublisher{American Mathematical Society},
\blocation{Providence, RI}
(\byear{2000}).
\doiurl{10.1090/conm/253/03936}
\end{bchapter}
\endbibitem

\bibitem[\protect\citeauthoryear{Nie and
  Schweighofer}{2007}]{NieSchweighofer2007}
\begin{barticle}
\bauthor{\bsnm{Nie}, \binits{J.}},
\bauthor{\bsnm{Schweighofer}, \binits{M.}}:
\batitle{On the complexity of {P}utinar's {P}ositivstellensatz}.
\bjtitle{Journal of Complexity}
\bvolume{23}(\bissue{1}),
\bfpage{135}--\blpage{150}
(\byear{2007})
\doiurl{10.1016/j.jco.2006.07.002}
\end{barticle}
\endbibitem

\bibitem[\protect\citeauthoryear{Bundfuss and D{\"u}r}{2009}]{Bundfuss2009}
\begin{barticle}
\bauthor{\bsnm{Bundfuss}, \binits{S.}},
\bauthor{\bsnm{D{\"u}r}, \binits{M.}}:
\batitle{An adaptive linear approximation algorithm for copositive programs}.
\bjtitle{SIAM Journal on Optimization}
\bvolume{20}(\bissue{1}),
\bfpage{30}--\blpage{53}
(\byear{2009})
\doiurl{10.1137/070711815}
\end{barticle}
\endbibitem

\bibitem[\protect\citeauthoryear{Y{\i}ld{\i}r{\i}m}{2012}]{Yildirim2012}
\begin{barticle}
\bauthor{\bsnm{Y{\i}ld{\i}r{\i}m}, \binits{E.A.}}:
\batitle{On the accuracy of uniform polyhedral approximations of the copositive
  cone}.
\bjtitle{Optimization Methods and Software}
\bvolume{27}(\bissue{1}),
\bfpage{155}--\blpage{173}
(\byear{2012})
\doiurl{10.1080/10556788.2010.540014}
\end{barticle}
\endbibitem

\bibitem[\protect\citeauthoryear{Gvozdenovi{\'c} and
  Laurent}{2007}]{Gvozdenovic2007}
\begin{barticle}
\bauthor{\bsnm{Gvozdenovi{\'c}}, \binits{N.}},
\bauthor{\bsnm{Laurent}, \binits{M.}}:
\batitle{Semidefinite bounds for the stability number of a graph via sums of
  squares of polynomials}.
\bjtitle{Mathematical Programming}
\bvolume{110}(\bissue{1}),
\bfpage{145}--\blpage{173}
(\byear{2007})
\doiurl{10.1007/s10107-006-0062-8}
\end{barticle}
\endbibitem

\bibitem[\protect\citeauthoryear{Dickinson}{2010}]{Dickinson2010}
\begin{barticle}
\bauthor{\bsnm{Dickinson}, \binits{P.J.C.}}:
\batitle{An improved characterisation of the interior of the completely
  positive cone}.
\bjtitle{Electronic Journal of Linear Algebra}
\bvolume{20},
\bfpage{723}--\blpage{729}
(\byear{2010})
\doiurl{10.13001/1081-3810.1404}
\end{barticle}
\endbibitem

\bibitem[\protect\citeauthoryear{Laurent and Vargas}{2022}]{LaurentVargas2205}
\begin{barticle}
\bauthor{\bsnm{Laurent}, \binits{M.}},
\bauthor{\bsnm{Vargas}, \binits{L.F.}}:
\batitle{On the exactness of sum-of-squares approximations for the cone of
  $5\times5$ copositive matrices}.
\bjtitle{Linear Algebra and its Applications}
\bvolume{651},
\bfpage{26}--\blpage{50}
(\byear{2022})
\doiurl{10.1016/j.laa.2022.06.015}
{\href{https://arxiv.org/abs/2205.05381}{{arXiv:2205.05381}}}.
\bcomment{arXiv:2205.05381}
\end{barticle}
\endbibitem

\bibitem[\protect\citeauthoryear{Peyrl and Parrilo}{2008}]{PeyrlParrilo2008}
\begin{barticle}
\bauthor{\bsnm{Peyrl}, \binits{H.}},
\bauthor{\bsnm{Parrilo}, \binits{P.A.}}:
\batitle{Computing sum of squares decompositions with rational coefficients}.
\bjtitle{Theoretical Computer Science}
\bvolume{409}(\bissue{2}),
\bfpage{269}--\blpage{281}
(\byear{2008})
\doiurl{10.1016/j.tcs.2008.09.025}
\end{barticle}
\endbibitem

\bibitem[\protect\citeauthoryear{Kaltofen et~al.}{2012}]{Kaltofen2012}
\begin{barticle}
\bauthor{\bsnm{Kaltofen}, \binits{E.L.}},
\bauthor{\bsnm{Li}, \binits{B.}},
\bauthor{\bsnm{Yang}, \binits{Z.}},
\bauthor{\bsnm{Zhi}, \binits{L.}}:
\batitle{Exact certification in global polynomial optimization via
  sums-of-squares of rational functions with rational coefficients}.
\bjtitle{Journal of Symbolic Computation}
\bvolume{47}(\bissue{1}),
\bfpage{1}--\blpage{15}
(\byear{2012})
\doiurl{10.1016/j.jsc.2011.08.002}
\end{barticle}
\endbibitem

\bibitem[\protect\citeauthoryear{Magron and {Safey El Din}}{2021}]{Magron2021}
\begin{barticle}
\bauthor{\bsnm{Magron}, \binits{V.}},
\bauthor{\bsnm{{Safey El Din}}, \binits{M.}}:
\batitle{On exact {R}eznick, {H}ilbert--{A}rtin and {P}utinar's
  representations}.
\bjtitle{Journal of Symbolic Computation}
\bvolume{107},
\bfpage{221}--\blpage{250}
(\byear{2021})
\doiurl{10.1016/j.jsc.2021.03.005}
\end{barticle}
\endbibitem

\bibitem[\protect\citeauthoryear{Tura et~al.}{2018}]{Tura2018}
\begin{barticle}
\bauthor{\bsnm{Tura}, \binits{J.}},
\bauthor{\bsnm{Aloy}, \binits{A.}},
\bauthor{\bsnm{Quesada}, \binits{R.}},
\bauthor{\bsnm{Lewenstein}, \binits{M.}},
\bauthor{\bsnm{Sanpera}, \binits{A.}}:
\batitle{Separability of diagonal symmetric states: a quadratic conic
  optimization problem}.
\bjtitle{Quantum}
\bvolume{2},
\bfpage{45}
(\byear{2018})
\doiurl{10.22331/q-2018-01-12-45}
{\href{https://arxiv.org/abs/1706.09423}{{arXiv:1706.09423}}}.
\bcomment{arXiv:1706.09423}
\end{barticle}
\endbibitem

\bibitem[\protect\citeauthoryear{Doherty et~al.}{2002}]{Doherty2002}
\begin{barticle}
\bauthor{\bsnm{Doherty}, \binits{A.C.}},
\bauthor{\bsnm{Parrilo}, \binits{P.A.}},
\bauthor{\bsnm{Spedalieri}, \binits{F.M.}}:
\batitle{Distinguishing separable and entangled states}.
\bjtitle{Physical Review Letters}
\bvolume{88}(\bissue{18}),
\bfpage{187904}
(\byear{2002})
\doiurl{10.1103/PhysRevLett.88.187904}
\end{barticle}
\endbibitem

\bibitem[\protect\citeauthoryear{Doherty et~al.}{2004}]{Doherty2004}
\begin{barticle}
\bauthor{\bsnm{Doherty}, \binits{A.C.}},
\bauthor{\bsnm{Parrilo}, \binits{P.A.}},
\bauthor{\bsnm{Spedalieri}, \binits{F.M.}}:
\batitle{Complete family of separability criteria}.
\bjtitle{Physical Review A}
\bvolume{69}(\bissue{2}),
\bfpage{022308}
(\byear{2004})
\doiurl{10.1103/PhysRevA.69.022308}
\end{barticle}
\endbibitem

\bibitem[\protect\citeauthoryear{Shen and Zhong}{2026}]{PaperII}
\begin{botherref}
\oauthor{\bsnm{Shen}, \binits{J.}},
\oauthor{\bsnm{Zhong}, \binits{H.}}:
Bosonic symmetric extensions: resolving the star versus complete-graph
  question, with an exact two-level separation of the {PPT}-{DPS} hierarchy.
Companion paper, in preparation
(2026)
\end{botherref}
\endbibitem

\bibitem[\protect\citeauthoryear{Lov{\'a}sz}{1979}]{Lovasz1979}
\begin{barticle}
\bauthor{\bsnm{Lov{\'a}sz}, \binits{L.}}:
\batitle{On the {S}hannon capacity of a graph}.
\bjtitle{IEEE Transactions on Information Theory}
\bvolume{25}(\bissue{1}),
\bfpage{1}--\blpage{7}
(\byear{1979})
\doiurl{10.1109/TIT.1979.1055985}
\end{barticle}
\endbibitem

\bibitem[\protect\citeauthoryear{Schrijver}{1979}]{Schrijver1979}
\begin{barticle}
\bauthor{\bsnm{Schrijver}, \binits{A.}}:
\batitle{A comparison of the {D}elsarte and {L}ov{\'a}sz bounds}.
\bjtitle{IEEE Transactions on Information Theory}
\bvolume{25}(\bissue{4}),
\bfpage{425}--\blpage{429}
(\byear{1979})
\doiurl{10.1109/TIT.1979.1056072}
\end{barticle}
\endbibitem

\bibitem[\protect\citeauthoryear{Gr{\"o}tschel et~al.}{1981}]{Grotschel1981}
\begin{barticle}
\bauthor{\bsnm{Gr{\"o}tschel}, \binits{M.}},
\bauthor{\bsnm{Lov{\'a}sz}, \binits{L.}},
\bauthor{\bsnm{Schrijver}, \binits{A.}}:
\batitle{The ellipsoid method and its consequences in combinatorial
  optimization}.
\bjtitle{Combinatorica}
\bvolume{1}(\bissue{2}),
\bfpage{169}--\blpage{197}
(\byear{1981})
\doiurl{10.1007/BF02579273}
\end{barticle}
\endbibitem

\bibitem[\protect\citeauthoryear{Pe{\~n}a et~al.}{2007}]{Pena2007}
\begin{barticle}
\bauthor{\bsnm{Pe{\~n}a}, \binits{J.}},
\bauthor{\bsnm{Vera}, \binits{J.}},
\bauthor{\bsnm{Zuluaga}, \binits{L.F.}}:
\batitle{Computing the stability number of a graph via linear and semidefinite
  programming}.
\bjtitle{SIAM Journal on Optimization}
\bvolume{18}(\bissue{1}),
\bfpage{87}--\blpage{105}
(\byear{2007})
\doiurl{10.1137/05064401X}
\end{barticle}
\endbibitem

\bibitem[\protect\citeauthoryear{Dukanovic and Rendl}{2010}]{Dukanovic2010}
\begin{barticle}
\bauthor{\bsnm{Dukanovic}, \binits{I.}},
\bauthor{\bsnm{Rendl}, \binits{F.}}:
\batitle{Copositive programming motivated bounds on the stability and the
  chromatic numbers}.
\bjtitle{Mathematical Programming}
\bvolume{121}(\bissue{2}),
\bfpage{249}--\blpage{268}
(\byear{2010})
\doiurl{10.1007/s10107-008-0233-x}
\end{barticle}
\endbibitem

\bibitem[\protect\citeauthoryear{Peres}{1996}]{Peres1996}
\begin{barticle}
\bauthor{\bsnm{Peres}, \binits{A.}}:
\batitle{Separability criterion for density matrices}.
\bjtitle{Physical Review Letters}
\bvolume{77}(\bissue{8}),
\bfpage{1413}--\blpage{1415}
(\byear{1996})
\doiurl{10.1103/PhysRevLett.77.1413}
\end{barticle}
\endbibitem

\bibitem[\protect\citeauthoryear{Horodecki et~al.}{1996}]{Horodecki1996}
\begin{barticle}
\bauthor{\bsnm{Horodecki}, \binits{M.}},
\bauthor{\bsnm{Horodecki}, \binits{P.}},
\bauthor{\bsnm{Horodecki}, \binits{R.}}:
\batitle{Separability of mixed states: necessary and sufficient conditions}.
\bjtitle{Physics Letters A}
\bvolume{223}(\bissue{1--2}),
\bfpage{1}--\blpage{8}
(\byear{1996})
\doiurl{10.1016/S0375-9601(96)00706-2}
\end{barticle}
\endbibitem

\bibitem[\protect\citeauthoryear{Doherty et~al.}{2005}]{Doherty2005}
\begin{barticle}
\bauthor{\bsnm{Doherty}, \binits{A.C.}},
\bauthor{\bsnm{Parrilo}, \binits{P.A.}},
\bauthor{\bsnm{Spedalieri}, \binits{F.M.}}:
\batitle{Detecting multipartite entanglement}.
\bjtitle{Physical Review A}
\bvolume{71}(\bissue{3}),
\bfpage{032333}
(\byear{2005})
\doiurl{10.1103/PhysRevA.71.032333}
\end{barticle}
\endbibitem

\bibitem[\protect\citeauthoryear{Terhal et~al.}{2003}]{Terhal2003}
\begin{barticle}
\bauthor{\bsnm{Terhal}, \binits{B.M.}},
\bauthor{\bsnm{Doherty}, \binits{A.C.}},
\bauthor{\bsnm{Schwab}, \binits{D.}}:
\batitle{Symmetric extensions of quantum states and local hidden variable
  theories}.
\bjtitle{Physical Review Letters}
\bvolume{90}(\bissue{15}),
\bfpage{157903}
(\byear{2003})
\doiurl{10.1103/PhysRevLett.90.157903}
\end{barticle}
\endbibitem

\bibitem[\protect\citeauthoryear{Navascu{\'e}s et~al.}{2009}]{Navascues2009}
\begin{barticle}
\bauthor{\bsnm{Navascu{\'e}s}, \binits{M.}},
\bauthor{\bsnm{Owari}, \binits{M.}},
\bauthor{\bsnm{Plenio}, \binits{M.B.}}:
\batitle{Power of symmetric extensions for entanglement detection}.
\bjtitle{Physical Review A}
\bvolume{80}(\bissue{5}),
\bfpage{052306}
(\byear{2009})
\doiurl{10.1103/PhysRevA.80.052306}
\end{barticle}
\endbibitem

\bibitem[\protect\citeauthoryear{Werner}{1989}]{Werner1989}
\begin{barticle}
\bauthor{\bsnm{Werner}, \binits{R.F.}}:
\batitle{Quantum states with {E}instein--{P}odolsky--{R}osen correlations
  admitting a hidden-variable model}.
\bjtitle{Physical Review A}
\bvolume{40}(\bissue{8}),
\bfpage{4277}--\blpage{4281}
(\byear{1989})
\doiurl{10.1103/PhysRevA.40.4277}
\end{barticle}
\endbibitem

\bibitem[\protect\citeauthoryear{Horodecki et~al.}{2009}]{Horodecki2009}
\begin{barticle}
\bauthor{\bsnm{Horodecki}, \binits{R.}},
\bauthor{\bsnm{Horodecki}, \binits{P.}},
\bauthor{\bsnm{Horodecki}, \binits{M.}},
\bauthor{\bsnm{Horodecki}, \binits{K.}}:
\batitle{Quantum entanglement}.
\bjtitle{Reviews of Modern Physics}
\bvolume{81}(\bissue{2}),
\bfpage{865}--\blpage{942}
(\byear{2009})
\doiurl{10.1103/RevModPhys.81.865}
\end{barticle}
\endbibitem

\bibitem[\protect\citeauthoryear{G{\"u}hne and T{\'o}th}{2009}]{Guhne2009}
\begin{barticle}
\bauthor{\bsnm{G{\"u}hne}, \binits{O.}},
\bauthor{\bsnm{T{\'o}th}, \binits{G.}}:
\batitle{Entanglement detection}.
\bjtitle{Physics Reports}
\bvolume{474}(\bissue{1--6}),
\bfpage{1}--\blpage{75}
(\byear{2009})
\doiurl{10.1016/j.physrep.2009.02.004}
\end{barticle}
\endbibitem

\bibitem[\protect\citeauthoryear{Shen and Zhong}{2026}]{CertArchive2026}
\begin{botherref}
\oauthor{\bsnm{Shen}, \binits{J.}},
\oauthor{\bsnm{Zhong}, \binits{H.}}:
Certificate archive and standalone verifier for ``Explicit Separators for
  Consecutive Levels of {P}arrilo's Sum-of-Squares Hierarchy over the
  Copositive Cone''.
Zenodo.
Source repository:
  \url{https://github.com/Mercury0828/parrilo-copositive-separators}
(2026).
\doiurl{10.5281/zenodo.22134333}
\end{botherref}
\endbibitem

\end{thebibliography}
\end{document}